\documentclass[nofootinbib, a4paper, aps, pra, reprint]{revtex4-1}

\usepackage[T1]{fontenc}
\usepackage[english]{babel}
\usepackage[dvipsnames]{xcolor}
\usepackage{amsfonts, amsmath, amssymb, amsthm} 
\usepackage{dsfont, braket, xfrac, fancyvrb} 
\usepackage{hyperref, setspace, enumitem, titlesec, tikz} 
\usepackage{bbold, bm, bbm}

\renewcommand\bra[1]{{\langle{#1}|}}
\makeatletter
\renewcommand\ket[1]{%
  \@ifnextchar\bra{\k@t{#1}\!}{\k@t{#1}}%
}
\newcommand\k@t[1]{{|{#1}\rangle}}
\makeatother

\addto\captionsenglish{}
\addto\captionsenglish{}

\newtheorem{lemma}{Lemma}
\newcommand{\FC}{\mathcal{F}}
\newcommand{\HC}{\mathcal{H}}
\newcommand{\OC}{\mathcal{O}}

\hypersetup{colorlinks,
            linkcolor = {blue!50!black},
            citecolor = {blue!50!black},
            urlcolor = {blue!80!black}
            }

\titleformat{\section}{\small\bfseries\uppercase}{\thesection.}{0.7em}{\centering}

\DeclareMathOperator{\tr}{tr}
 
\begin{document}
\title{Biased Random Access Codes}
\author{Gabriel Pereira Alves}
\affiliation{Faculty of Physics, University of Warsaw, Pasteura 5, 02-093 Warsaw, Poland}
\author{Nicolas Gigena}
\affiliation{Faculty of Physics, University of Warsaw, Pasteura 5, 02-093 Warsaw, Poland}
\author{J\k{e}drzej Kaniewski}
\affiliation{Faculty of Physics, University of Warsaw, Pasteura 5, 02-093 Warsaw, Poland}
\date{\today}

\begin{abstract}
A random access code (RAC) is a communication task in which the sender encodes a random message into a shorter one to be decoded by the receiver so that a randomly chosen character of the original message is recovered with some probability. Both the message and the character to be recovered are assumed to be uniformly distributed. In this paper, we extend this protocol by allowing more general distributions of these inputs, which alters the encoding and decoding strategies optimizing the protocol performance, with either classical or quantum resources. We approach the problem of optimizing the performance of these biased RACs with both numerical and analytical tools. On the numerical front, we present algorithms that allow a numerical evaluation of the optimal performance over both classical and quantum strategies and provide a Python package designed to implement them, called RAC-tools. We then use this numerical tool to investigate single-parameter families of biased RACs in the $n^2 \mapsto 1$ and $2^d \mapsto 1$ scenarios. For RACs in the $n^2 \mapsto 1$ scenario, we derive a general upper bound for the cases in which the inputs are not correlated, which coincides with the quantum value for $n=2$ and, in some cases for $n=3$. Moreover, it is shown that attaining this upper bound self-tests pairs or triples of rank-1 projective measurements, respectively. An analogous upper bound is derived for the value of RACs in the $2^d \mapsto 1$ scenario, which is shown to be always attainable using mutually unbiased measurements if the distribution of input strings is unbiased.
\end{abstract}

\maketitle

\clearpage

\section{Introduction}

In the past decades several instances have been found in which quantum resources provide an advantage in the performance of a given task. Quantum computing algorithms \cite{Deutsch92, Grover96, Bernstein97, Simon97, Montanaro16}, such as Shor's factorization algorithm \cite{Shor94}, are just one example of the power of quantum resources: Spatially separated parties can use entanglement \cite{Horodecki09} in a shared quantum state, for instance, to improve their performance in a nonlocal game \cite{Brunner14, Uola20, Budroni22}, to quantum teleport \cite{Bennett93, Bouwmeester97, Boschi98, Pirandola15} the state of a third system held by one of them, or to densely encode classical information to be sent via a quantum channel \cite{Bennett92, Mattle96, Nielsen11}. Quantum devices have also shown to be powerful resources for certain communication tasks in which a quantum state is prepared by one party and sent to another one, who performs a measurement to extract information. Such tasks are known as prepare-and-measure experiments, and they find application in quantum information processing protocols like quantum key distribution (QKD) \cite{Bennett84, Inoue02, Grosshans03, Stucki05, Xu20}, randomness certification \cite{Li12, Bowles14, Passaro15, Acin16}, and quantum random access codes \cite{Ambainis99, Nayak99, Ambainis02, Ambainis09}. A random access code (RAC) is a communication task in which a string of characters, chosen at random from a given alphabet, is encoded into a shorter string in such a way that any of the characters in the original string can be recovered, with some probability, by means of a decoding strategy. Both the string to be encoded and the character to be recovered are uniformly distributed, with the encoding party not knowing in advance which character should be retrieved by the decoding procedure. In that sense, the RAC can be understood as a form of nondeterministic data compression.

The implementation of the RAC protocol and its variations have been the subject of intense research, finding applications in cryptography \cite{Pawlowski11, Chaturvedi21}, self-testing of measurements \cite{Tavakoli18, Farkas19}, foundational aspects of no-signaling correlations \cite{Pawlowski09}, and quantum communication complexity \cite{Buhrman01, Gavinsky09}. In this work, we introduce a generalization of the RAC protocol in which neither the string to be encoded, nor the character to be recovered are uniformly distributed. Biasing the distribution of inputs has a nontrivial effect on the encoding and decoding strategies optimizing the performance of the protocol in both, quantum and classical realizations. We approach here the problem of finding the optimal performance of such biased RACS, and the strategies attaining it, with both numerical and analytical techniques. In the numerical front we present the RAC-tools Python package, built to implement algorithms providing the exact classical value and lower bounds to the quantum value of an arbitrary biased RAC. On the analytical side, we derive upper bounds for the optimal performance over projective measurements of RACs in which the character strings to be encoded consist either of two characters to be chosen from length $d$ alphabet, or of $n$ characters to be chosen from a length $2$ alphabet. In the cases in which these upper bounds are attainable, we study the optimal quantum strategies achieving the optimal performance, paying special attention to the dependence of the optimal measurements on the biasing parameters and the regions in parameter space in which quantum strategies provide an advantage over classical ones. These analytical results are then compared with those produced by the numerical package.


\section{Biased RACs \label{s:b-racs}}\label{s:b-racs}

A RAC scenario, as depicted in FIG.~\ref{RAC}, involves two parties, Alice and Bob. A RAC scenario is parameterized by integers $n, m, d$,  which we assume to be equal to or larger than 2. Alice is given an $n$-character string ${\bf x}\in S=\{0, ..., m-1\}^{\times n}$ and asked to encode it into a single character $\mu\in\{0, ..., d-1\}$, which she will later send to Bob. Bob, on the other hand, is asked to decode Alice's message in order to retrieve the value of the $y$-th character $x_y$ in the original string. In order to do so Bob evaluates the image $b\in\{0, ..., m-1\}$ of $\mu$ under a previously chosen decoding function, represented in the figure by Bob's box. We consider the task successful when Bob correctly guesses $x_y$ from his decoding function, i.e., when $b=x_y$. We denote a scenario like the one described above by $\smash{n^m\overset{d}{\mapsto} 1}$ RAC, where we allow the message and the characters in the string to belong to alphabets with different cardinalities, $m\neq d$. Whenever $m=d$, however, we will use the notation $n^d\mapsto 1$, since it is the usual notation in the literature.
\begin{figure}[t!]
\begin{center}
    \includegraphics{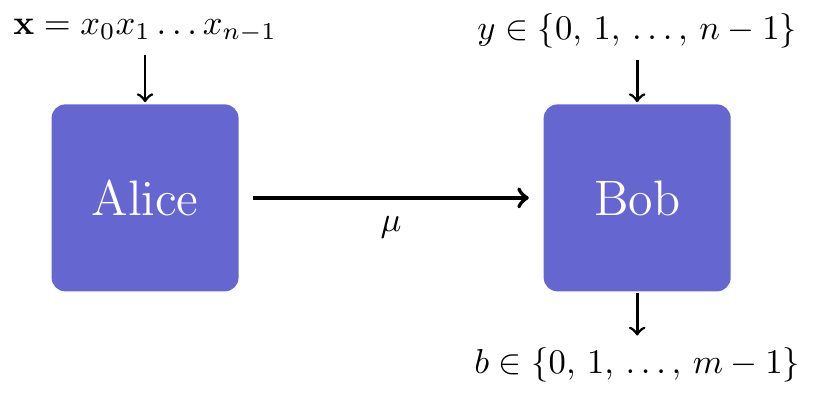}
\end{center}
\caption{The $n^m\overset{d}{\mapsto} 1$ RAC. Alice encodes her input $\mathbf{x} = x_0 \, x_1 \, \dots \, x_{n - 1}$ into a message $\mu$ which is sent to Bob. Based on the message $\mu$ and his input $y$, Bob tries to guess the $y$-th character of Alice. Each character $x_i$ of $\mathbf{x}$ ranges from 0 to $m - 1$, for $i = 0,~1,~ ...,~ n - 1$.} \label{RAC}
\end{figure}

The figure of merit that is commonly used to study a RAC is the \emph{average success probability} $\bar{P}$, which is simply an average of the winning probabilities over all combinations of $\mathbf{x}$ and $y$:
\begin{align}
\bar{P} = \frac{1}{nm^n} \sum_{\mathbf{x}, y} p (b = x_y \, | \, \mathbf{x}, ~ y), \label{ASP}
\end{align}
where $p (b=x_y \, | \, \mathbf{x}, ~ y)$ denotes the probability of a successful decoding when string ${\bf x}$ is encoded and character $x_y$ is to be recovered, and the factor $\frac{1}{n m^n}$ reflects the assumption that both $\bf x$ and $y$ are uniformly distributed. The optimal pairs of encoding-decoding strategies implemented by Alice and Bob will therefore be those maximizing $\bar P$. Note that the trivial strategy of outputting a fixed (or random) value of $b$ achieves average success probability of $\frac{1}{m}$, so we will be interested only in strategies that outperform this value.

In this work, we study a more general class of RACs in which the distribution of Alice's and Bob's inputs, $\bf x$ and $y$ respectively, is not necessarily uniform. We refer to these as \emph{biased} RACs, or $b$-RACs.
In fact, let us start with the most general real linear functional, i.e.~a tensor $\alpha_{\mathbf{x} y b}$ of order ($n + 2$), which attributes a specific weight to each combination of inputs $\mathbf{x}$ and $y$ and output $b$. The value of this functional on a probability distribution $p (b \, | \, \mathbf{x},\, y)$ equals:
\begin{align}
\mathcal{F} = \sum_{\mathbf{x}, y, b} \alpha_{\mathbf{x} y b} \, p (b \, | \, \mathbf{x},\, y). \label{general functional}
\end{align}
Since the probability distribution is normalized an additive shift in the coefficients of the tensor results in an additive shift in the value. Thus, we can focus on tensors which are non-negative. Similarly, by normalizing the coefficients we can, without loss of generality, consider tensors satisfying $\sum_{\mathbf{x}, y, b} \alpha_{\mathbf{x} y b} = 1$. Now, we would like to focus our attention on functionals that to some extent resemble the standard random access code, where the winning condition reads $b = x_{y}$. In this spirit, we will consider only tensors whose coefficients vanish whenever $b \neq x_{y}$. Due to non-negativity and normalization such linear functionals can be interpreted as probability distributions over ${\bf x}$ and $y$, and let us denote them by $\alpha_{{\bf x} y}$. Then, the value of the functional equals
\begin{align}
\mathcal{F} = \sum_{\mathbf{x}, y} \alpha_{\mathbf{x} y} \, p ( b=x_y \, | \, \mathbf{x},\, y), \label{RAC functional}
\end{align}
and it should be clear that this is precisely the same as the original RAC except that the distribution of inputs might be nonuniform (in the RAC we have $\alpha_{{\bf x} y}=\frac{1}{nm^n}$). Hence, biasing the input distribution of a RAC is analogous to biasing or tilting functionals in a Bell nonlocality scenario \cite{Acin12, Yang13, Bamps15}, as the consequence of adopting bias is to modify the functional in Eq.~\eqref {ASP} and, in turn, the optimal realization. An interesting aspect of this generalization is that the inputs of Alice and Bob are not necessarily independent. To the best of our knowledge such scenarios were first analyzed in Ref.~\cite{Kaczynska21}. An alternative direction would be to change the goal of the decoding function from recovering a given character to a more general function of the input string, a generalization that has recently been explored in Ref.~\cite{Doriguello21}, but we do not consider such scenarios in this work.


As is the case with Bell scenarios, $b$-RACs can be fundamentally interpreted as an experiment in which a particular behavior of the involved devices, specified by the conditional probabilities $\{p(b=x_y|{\bf x}, y)\}$, determines a certain value for the previously specified figure of merit $\FC$ in Eq.\eqref{RAC functional}. In the case of $b$-RACs the behavior is determined by the encoding and decoding strategies implemented by Alice and Bob. So far we have focused on classical strategies in which, upon receiving her input $\bf x$, Alice computes its image $\mu$ under an encoding function $E : \{ 0,\, ...,\, m - 1 \}^{\times n} \mapsto \{ 0,\, ...,\, d - 1 \}$ and sends it to Bob. Bob takes $\mu$ as the argument of his $y$-th decoding function $D_y : \{ 0,\, ...,\, d - 1 \}  \mapsto \{ 0,\, ...,\, m - 1 \}$ producing $b = D_y(\mu)$. The strategies we just described are deterministic, i.e., the probability of Bob's output being $b$ is given by
\begin{align}
p(b \, |\, \mathbf{x}, \, y) = 
\begin{cases}
    1 & \text{if } b = D_y(E(\mathbf{x})) \\
    0 & \text{otherwise} \label{deterministic behaviors}
\end{cases}, \quad \forall ~\mathbf{x}, \, y, \, b.
\end{align}
Clearly, if Alice and Bob decide to employ some nondeterministic strategy (even if we allow them to share classical randomness), the corresponding behavior $\{p(b=x_y|{\bf x}, y)\}$ will belong to the convex hull of classical deterministic behaviors. It follows then that the $b$-RAC functional in Eq.~\eqref{RAC functional} attains its maximum $\FC_C$ for one of these deterministic strategies, which can be found through an exhaustive search, as is the case for the local value of Bell functionals. It should nonetheless be noted that in a bipartite scenario with $n$ measurement settings per party and $m$ measurements per setting the number of deterministic strategies scales as $m^{2n}$, whereas in the $\smash{n^m\overset{d}{\mapsto} 1}$ $b$-RAC scenario this number grows double-exponentially, according to $d^{m^n}\times m^{dn}$. Therefore, computing the classical value through exhaustive search becomes infeasible even for small values of $n$, $d$, and $m$. The following lemma provides a simpler and more efficient approach to the computation of the classical value.

\begin{lemma}
\textup{(a)} For a $n^m\overset{d}{\mapsto} 1$ b-RAC and a fixed encoding function $E ( \mathbf{x} ) = \mu$, the optimal decoding functions are $D_y^*(\mu) = b$, where $b$ is the character that maximizes the sum
\begin{align}
\sum_{\mathbf{x}\in S} \alpha_{\mathbf{x}, y, b} \, \delta_{\mu, E ( \mathbf{x} )}. \label{corollary 1 equation}
\end{align}
$\textup{(b)}$ For a $n^m\overset{d}{\mapsto} 1$ b-RAC and fixed decoding functions $\{D_y\}_y$, the optimal encoding function is given by $E^*(\mathbf{x}) = \mu$, where $\mu$ is the character that maximizes the sum
\begin{align}
\sum_{y = 0}^{n - 1} \alpha_{\mathbf{x}, y, D_y ( \mu )}. \
\end{align}
\label{EXHAUSTIVE SEARCH LEMMA}
\end{lemma}
It follows from Lemma \ref{EXHAUSTIVE SEARCH LEMMA}, which is proved in Appendix \ref{proof of the exhaustive search lemma}, that we can reduce the exhaustive search to a search either over encoding functions only -- statement (a) -- or over decoding functions only -- statement (b). In the first case, the complexity of the problem reduces from $d^{m^n}\times m^{dn}$ to $d^{m^n}\times nmd$, whereas in the second it reduces to $d\times m^{n(d + 1)}$. A similar treatment is usually implemented in Bell scenarios, where optimal measurement (response function in the classical case) can be computed if the remaining components are fixed.

\subsection{Quantum value of a biased RAC} \label{see-saw method}

As an information processing task, the RAC can be generalized to quantum strategies. A quantum strategy involves quantum devices, which transforms the $b$-RAC scenario into a particular case of a prepare-and-measure experiment \cite{Gallego10}. Such scenarios are often encountered in many new quantum technologies such as quantum communication \cite{Gisin07} and quantum cryptography \cite{Gisin02}. In a quantum strategy Alice, upon receiving a string $\mathbf{x}$, encodes it no longer in a classical character $\mu$, but in a qudit density operator $\rho_\mathbf{x}$ over a $d$-dimensional Hilbert space $\HC$ that she then sends to Bob. Once Bob receives Alice's preparation and his input $y$, he performs a decoding measurement described by operators $\{ M_y^b \}_{b = 1}^m$, producing output $b$. From Born's rule it follows then that the probability of a successful decoding is $p(b=x_y|{\bf x}, y)=\tr (\rho_{\bf x} M_y^{x_y})$, and the ensuing value of the figure of merit  in Eq.~\eqref{RAC functional} reads
\begin{equation}
\mathcal{F} = \sum_{\mathbf{x}, y} \alpha_{\mathbf{x} y} \tr \, (\rho_\mathbf{x} M_y^{x_y}). \label{qrac-functional}
\end{equation}
The optimal quantum encoding-decoding strategies will be those for which the functional in Eq.~\eqref{qrac-functional} attains its maximum value, which we will denote by $\FC_Q$ to distinguish it from the optimal value over the set of classical behaviors, which we will denote from now on by $\FC$. The optimization problem over both preparations $\rho_{\bf x}$ and measurements $\{M_y^{x_y}\}$ involved in the determination of $\FC_Q$ is in general hard, but it can be approached numerically by means of a {\it see-saw} algorithm, as we will describe below. This method, which relies on the fact that for fixed preparations optimal measurements can be found efficiently and vice-versa, is a known numerical technique for obtaining lower bounds for the quantum value of Bell inequalities \cite{Werner01, Ito06}.

In a nutshell, the see-saw algorithm consists in the repeated implementation of a two-step optimization procedure, since we need to optimize the $b$-RAC functional over preparations and over measurements. We start with a set of randomly chosen measurements, for which we can find the optimal preparations by noting that the functional $\FC$ can be written as \footnote{We relax here the RAC condition $b=x_y$ to highlight the fact that the argument holds as it is for more general functionals.}
\begin{align}
\mathcal{F} = \sum_{\mathbf{x}, y ,b} \alpha_{\mathbf{x} y b} \tr \, (\rho_\mathbf{x} M_y^b)
= \sum_\mathbf{x} \tr \, ( \rho_\mathbf{x} \sum_{y, b} \alpha_{\mathbf{x} y b} M_y^b ), \label{optimization of states}
\end{align}
becoming thus apparent that the density matrix $\rho_\mathbf{x}$ maximizing the trace is given by a state associated with the largest eigenvalue of the positive semidefinite operator $\sum_{y, b} \alpha_{\mathbf{x}yb} M_y^b$. This means that if we are interested in computing the quantum value we can without loss of generality assume $\rho_\mathbf{x}$ to be a pure state. Once the preparations are determined the algorithm proceeds to the second step, which is finding the ensuing optimal measurements for these preparations. In order to do so it is convenient to again rewrite the functional $\FC$ as
\begin{align}
\mathcal{F} = \tr \, \bigg( \sum_{\mathbf{x}, y, b} M_y^b \alpha_{\mathbf{x}yb} \rho_\mathbf{x} \bigg)
= \tr \, \bigg( \sum_{y, b} M_y^b  \varrho_{y, b}  \bigg), \label{optimization of measurements}
\end{align}
where
\begin{equation}
\varrho_{y, b} : = \sum_\mathbf{x} \alpha_{\mathbf{x}yb} \rho_\mathbf{x}    
\end{equation}
 is a subnormalized density matrix. It is then easily seen that the optimization over measurements takes the form of a semidefinite program (SDP)
\begin{align}
\begin{aligned}
\max_{ \{M_b^y\} } & \quad \tr \, \Big( \textstyle\sum_b M_y^b  \varrho_{y, b}  \Big) \\
\text{s.t.}        & \quad M_y^b \succeq 0, \quad \forall \; b  \\
\text{and}         & \quad \textstyle\sum_{b} M_y^b = \mathds{1}, \label{SDP}
\end{aligned}
\end{align}
the solution of which, for all $n$ measurements, completes the second step in the optimization procedure and one iteration in the see-saw algorithm, with the newly found measurements becoming the starting point of the next iteration. Note that the problem in Eq.~\eqref{SDP}, is one of minimum-error discrimination \cite{Joonwoo15} of the states $\varrho_{y, b}$, i.e., the optimal measurements $\{ M_y^b \}_{b = 1}^m$ are those minimizing the error in the discrimination of $\varrho_{y, b}$, $\forall~b$.

Before closing this section some remarks are in order. First, it should be noted that while the iterative procedure described above will always converge in value to some maximum of $\FC$, there is no guarantee this is a global maximum. This will highly depend on the starting point, which is chosen at random. Second, it is worth noting that the see-saw procedure described above can be also implemented to find the optimal performance over classical strategies: Since advantage of quantum strategies is rooted in the possibility of performing incompatible measurements, restricting all measurement operators to be diagonal in the computational basis will reduce the optimization procedure to a maximization over classical strategies. This restriction is easily imposed by just initializing the algorithm with random diagonal matrices as seeds for measurements, which ensures that all the states and measurements arising during the see-saw procedure will be diagonal in the same basis. Note, however, that, unlike the exhaustive search previously discussed, this is a completely heuristic method. 

As a final observation, note that in the first step of the see-saw procedure, when the decoding strategy is fixed, the optimal value is given by
\begin{equation}
\max_{\{\rho_{\bf x}\}}\; \FC = 
\sum_{\bf x} \lambda_{\max} \left(\sum_{y, b}\alpha_{{\bf x} y b} M^b_y \right),
\label{opt_meas_only}
\end{equation}
where $\lambda_{\max} (\OC)$ denotes the largest eigenvalue of operator $\OC$. That is, finding the optimal preparations $\{ \rho_{\bf x} \}$ for a fixed set of measurements is an eigenvalue problem, which can be solved analytically, and therefore the we are left to look only for the optimal measurements. In spite of this simplification, finding the quantum value remains a hard problem in general. Nonetheless, as we will see later, there are cases in which it can be approached analytically.


\section{The RAC-tools Python package \label{S:package}} \label{S:package}

In this section, we introduce a Python package \cite{GitHub22} that implements the numerical methods described in the previous section. Our goal was to construct a tool that allows the user to easily determine basic properties of a $b$-RAC, such as its classical and quantum value. This tool requires some standard Python packages like \emph{numpy} and \emph{scipy}, as well as the \emph{cvxpy} package to solve SDPs. Amongst the different solvers that can be used with \emph{cvxpy}, we have found \emph{MOSEK} \cite{Andersen00, MOSEK} to be the most reliable. It is also available at no cost for academic use. Throughout this section, we provide a succinct description of how the package works, which is further elaborated in Appendix~\ref{Ap:pack}.

The RAC-tools package is written to implement both the exhaustive search and see-saw algorithms, which were discussed in section \ref{s:b-racs}. For the exhaustive search method, users can invoke the \verb+perform_search+ function, while the \verb+perform_seesaw+ function is employed for the see-saw optimization. Since our main focus is on biased RACs, an essential part of the package deals with the specification of the biasing tensor $\alpha_{{\bf x} y b}$, as defined by the functional in Eq.~\eqref{general functional}. Note that using this definition instead of the one in \eqref{RAC functional} allows the user to define functionals in a class larger than that of RAC functionals, which corresponds to the condition $b=x_y$. The desired biasing tensor can be written explicitly and passed to both\verb+perform_search+ and \verb+perform_seesaw+ in the form of a Python dictionary. Alternatively, the user can opt for any of the built-in biasing tensors provided by the package via the \verb+generate_bias+ function, and input only a reduced set of parameters.

Before entering into the more technical details of \verb+perform_search+ and \verb+perform_seesaw+, let us provide an example to better illustrate how \verb+generate_bias+ works. Consider a $2^2 \mapsto 1$ RAC in which the Alice's input $\bf x$ is uniformely distributed, but Bob is asked to retrieve the first character of $\bf x$ with probability $w \in [0, 1]$. The bias tensor defining this $b$-RAC is given by
\begin{align}
\alpha_{\mathbf{x} y} = 
\begin{cases}
    \frac{1}{4} w \, & \text{if} ~ y = 0, \\
    \frac{1}{4} (1-w)  & \text{otherwise},
\end{cases}
\end{align}
where the factor $\frac{1}{4}$ results from $\bf x$ being uniformly distributed. We can compute its classical and quantum value by passing to \verb+perform_search+ and \verb+perform_seesaw+, respectively, the string \verb+bias="Y_ONE"+ and the float \verb+weight=+$w$. While the variable \verb+weight+ encodes the amount of bias desired, the variable \verb+bias+ encodes the type of bias that the user wants to compute. For example, if the user passes \verb+weight=0.75+ as an argument, \verb+generate_bias+ builds a biasing tensor with components
\begin{align}
\alpha_{\mathbf{x} y} = 
\begin{cases}
    \frac{3}{16} \, & \text{if} ~ y = 0, \\
    \frac{1}{16}  & \text{otherwise}.
\end{cases}
\end{align}
In a similar manner to that demonstrated in the above example, we can also consider biases that exclusively affect Alice's input $\bf x$, or alternatively affect both $\bf x$ and $y$ simultaneously. RAC-tools includes several built-in bias families, including the one shown above, which we describe in more detail in Appendix~\ref{Ap:gen_fun}. Moving forward, we proceed now to the description of the operational details of the functions \verb+perform_search+ and \verb+perform_seesaw+.

\subsection{The \texorpdfstring{\protect\Verb+perform\_search+}{} function}

The main use of the function \verb+perform_search+ is to compute the optimal classical performance of an $\smash{n^m\overset{d}{\mapsto}} 1$ $b$-RAC. This function takes as argument the integers $n$, $d$, and $m$, encoded by the analogous Python parameters \verb+n+, \verb+d+, and \verb+m+, as well as the bias tensor $\alpha_{{\bf x} y b}$. The latter can be entered via the Python dictionary \verb+bias_tensor+ or as the aforementioned reduced set of parameters \verb+bias+ and \verb+weight+, for one of the built-in bias families. While these parameters fix a particular $b$-RAC scenario, the variable \verb+method+ defines the searching approach that should be employed by \verb+perform_search+. When setting \verb+method=1+ and \verb+method=2+, the corresponding implementations correspond to the approaches described in statements (a) and (b) of Lemma \ref{EXHAUSTIVE SEARCH LEMMA}, respectively. On the other hand, \verb+method=0+ implements a purely exhaustive search, where neither the encoding nor decoding functions are fixed. It is worth noting that the value of $m$ is set by default to be the same as that of $d$, and therefore there is no need to declare it when studying $n^d\mapsto 1$ $b$-RACs.

When \verb+perform_search+ finishes the execution, it generates a report like the one in FIG.~\ref{classical example}, for the case of the $2^2 \mapsto 1$ RAC. As can be seen in the figure, the report provides not only the optimal value of the functional but also an encoding-decoding pair attaining this value, along with information about computing time and total number of encoding-decoding functions. This report can be disabled by setting the variable \verb+verbose=False+, in which case \verb+perform_search+ still returns the information displayed in the report, but in the form of a dictionary, allowing the user to manipulate this information. A more detailed description of this function is given in Appendix~\ref{Ap:class_val_fun}.

\begin{figure}[ht!]
\centering
\footnotesize
\begin{BVerbatim}[commandchars=\\\{\}]
> perform_search(n=2, d=2, method=0)

==========================================================
                      RAC-tools v1.0
==========================================================

----------------- Summary of computation -----------------

Total time of computation: 0.001859 s
Total number of encoding/decoding functions: 256
Average time per function: 7e-06 s

----------- Analysis of the optimal realization ----------

Computation of the classical value for the 2\textsuperscript{2}-->1 RAC:
0.75
Number of functions achieving the computed value: 24

First functions found achieving the computed value
Encoding: 
E: [0, 0, 0, 1]
Decoding: 
D\textsubscript{0}: [0, 1]
D\textsubscript{1}: [0, 1]

------------------- End of computation -------------------
\end{BVerbatim}
\caption{Report produced by the function \protect\Verb+perform\_search+ for the unbiased $2^2 \mapsto 1$ RAC. In addition to the summary of the computation, the user is provided with the first strategy found attaining the optimal value and the number of such equivalent strategies. For the encoding function $E(\mathbf{x})$, the result is shown in a tuple organized in ascending order of $\mathbf{x}$, $[E(00...0), E(00...1), ..., E((m - 1)...(m - 1))]$. The decoding functions $D_y(\mu)$ follow a similar pattern: each row corresponds to a distinct input $y$, and the result is organized in ascending order of $\mu$.} 
\label{classical example}
\end{figure}

\subsection{The \texorpdfstring{\protect\Verb+perform\_seesaw+}{} function}

This function implements the see-saw algorithm as described in section \ref{s:b-racs}. As mentioned before, the algorithm always converges to a maximum of the functional described in Eq.~\eqref{qrac-functional}. However, it is not guaranteed that this maximum represents the global maximum due to the random initialization of the algorithm. Hence, in addition to providing the integers $n$, $d$, and $m$, along with the tensor $\alpha_{\mathbf{x} y b}$ that specifies the scenario, the user is required to enter the number of random initializations through the parameter \verb+seeds+ when invoking the function. It should be noted that since increasing the number of initializations increases the computation time, deciding which is the best value for \verb+seeds+ is a problem on its own; In Appendix~\ref{Ap:opt_val_fun} we provide, as a guide, TABLE~\ref{table:number of seeds}, which contains the number of seeds used for generating the numerical results presented in this work.

In addition, as noted in the description of the see-saw algorithm, it can also be used to compute the classical value of a $b$-RAC by restricting the measurements and preparations to be diagonal in the computational basis. This condition can be passed to the \verb+perform_seesaw+ function via the extra variable \verb+diagonal+. If \verb+diagonal=True+, the function initializes the see-saw algorithm with random diagonal measurements and the retrieved value corresponds to an estimation of the classical value. The default value of this variable is \verb+False+.

After finishing the computation, \verb+perform_seesaw+ prints a report including a short analysis of the measurement operators attaining the optimal value found, such as whether the operators are projective or if they are mutually unbiased. Moreover, the report informs also about the computation time and the number of random starting points used by the code. A detailed description about this data can be found in Appendix~\ref{Ap:opt_val_fun}. FIG.~\ref{simple example} shows an example of the report printed by the function in the case of the $2^2 \mapsto 1$ RAC. As before, this report can be disabled by setting the variable \verb+verbose=False+, in which case the information displayed in it is returned in the form of a Python dictionary.
    
\begin{figure}[t]
\centering
\footnotesize
\begin{BVerbatim}[commandchars=\\\{\}]
> perform_seesaw(n=2, d=2, seeds=5)

==========================================================
                      RAC-tools v1.0
==========================================================

----------------- Summary of computation -----------------

Number of random seeds: 5
Average time for each seed: 0.14852 s
Average number of iterations: 3
Seeds 1e-13 close to the best value: 5

----- Analysis of the optimal realization for seed #1 ----

Estimation of the quantum value for the 2\textsuperscript{2}-->1 RAC:
0.853553390593

Measurement operator ranks
M[0] ranks:  1  1
M[1] ranks:  1  1

Measurement operator projectiveness
M[0, 0]:  Projective		6.44e-15
M[0, 1]:  Projective		6.44e-15
M[1, 0]:  Projective		6.78e-15
M[1, 1]:  Projective		6.78e-15
 
Mutual unbiasedness of measurements
M[0] and M[1]:  MUB		5.91e-14

------------------- End of computation -------------------
\end{BVerbatim}
\caption{Report produced by the function \protect\Verb+perform\_seesaw+ for the unbiased $2^2 \mapsto 1$ RAC. In the first part of the report, the function produces a small summary of the computation, displaying information such as the number of random starting points, average time per starting point, etc. The subsequent part presents the optimal value found along with a short analysis of the measurements attaining this value. The notation \protect\Verb+M[y]+ refers to the $y$-th measurement, while \protect\Verb+M[y, b]+ refers to the operator yielding output $b$. In this way, the item \emph{Measurement operator ranks} presents the computed ranks for the operators of the $y$-th measurement, arranged in ascending order of $b$. Similarly, the item \emph{Measurement operator projectiveness} indicates whether the identified operators can be considered projective or not. The numerical value presented in the second column serves as a measure for projectiveness, with a value approaching zero indicating that this operator is close of being projective. Lastly, the item \emph{Mutual unbiasedness of measurements} analyzes the possibility of constructing each pair of measurements out of mutually unbiased bases. Analogous to the previous item, the second column provides a measure of proximity for a given pair. Further technical details regarding the \protect\Verb+perform\_seesaw+ function can be found in Appendix \ref{Ap:pack}.}
\label{simple example}
\end{figure}


\section{Analytical results for the \texorpdfstring{$\MakeLowercase{n}^{2}\mapsto 1$}{}  RAC \label{S:n23}}\label{S:n23}

As discussed in Section \ref{s:b-racs}, whether classical or quantum the strategies maximizing the $b$-RAC functional are in general hard to find analytically. An exception is provided by some RACs whose output is a single bit since, as we will see below, this greatly simplifies the solution of the corresponding optimization problems. Since we are interested in the advantage provided by quantum strategies, let us start by discussing the optimal classical performance.

\subsection{Classical}

We are now interested in finding the optimal encoding-decoding strategies for general biased $n^2\mapsto 1$ RAC. The case of unbiased RACs has been already studied in Ref.~\cite{Ambainis09}, and as we show below, the authors' analysis can be extended to the general case with slight modifications. 

A salient feature of the optimal classical encoding-decoding strategies for $b$-RACs is that in some cases they ignore part of the input. Excluding part of the input in the search for the optimal encoding and decoding strategies reduces the complexity of the problem, since it involves fewer bits; thus the optimal value is easier to compute if the subset of bits to be ignored is known beforehand. Unfortunately this knowledge does not seem to be available in advance, and in order to find the set of bits ignored by the best strategy we need to compare the values for all possible options, which makes the evaluation computationally hard. We will return to this issue when discussing optimal quantum strategies.

The following lemma shows how to find the optimal classical strategies under the assumption that no bit is ignored. If the actual optimal strategy for a given $b$-RAC does ignore part of the input string $\bf x$, the result applies to the set of bits taken into account by the strategy.

\begin{lemma}{ \label{lem_optclass}}

    The optimal nonbit-ignoring strategies for the $n^2\mapsto 1$ $b$-RAC comprise a weighted majority encoding function and identity map for decoding.
\end{lemma}

\begin{proof}
    We already know from Lemma \ref{EXHAUSTIVE SEARCH LEMMA} how to find the optimal decoding function for a given fixed encoding, and the optimal encoding function for a fixed decoding. In the particular case of the $n^2\mapsto 1$ $b$-RAC, given a fixed encoding $E:\{0, 1\}^{\times n}\mapsto {0, 1}$ mapping input $\bf x$ into a bit $\mu$, i.e., $E({\bf x})=\mu$, it implies that the optimal decoding function for the $y$-th must satisfy
    \begin{equation}
        D_y(\mu) = \begin{cases}
                    0 & \;\text{if}\; \sum_{{\bf x}}\alpha_{{\bf x}|y} \delta_{x_y 0} \geq \sum_{{\bf x}}\alpha_{{\bf x}|y} \delta_{x_y 1}, \\
                    1 &\; \text{otherwise}, \label{optdec}
                   \end{cases}
    \end{equation}
    where $\alpha_{\bf x|y}=\alpha_{{\bf x} y}/r_y$, with $r_y=\sum_{\bf x} \alpha_{{\bf x} y}$, can be interpreted as the probability of string $\bf x$ being Alice's input given that the $y$-th bit is to be recovered by Bob.

    Because both $D_y(\mu)$ and its argument $\mu$ are bits there are only four possible decoding functions: two constant maps $D_y(\mu)=0, 1$, the identity map $D_y(\mu)=\mu$, and a flip of the input $D_y(\mu)=1-\mu$. Now suppose the optimal decoding in Eq.~\eqref{optdec} corresponds to one of the constant functions, e.g., $D_y(\mu)\equiv 0$. This implies $\sum_{{\bf x}}\alpha_{{\bf x}|y} \delta_{x_y 0} \geq \sum_{{\bf x}}\alpha_{{\bf x}|y} \delta_{x_y 1}$ for both $\mu=0, 1$. Because $\sum_{{\bf x}}\alpha_{{\bf x}|y} \delta_{x_y 0} + \sum_{{\bf x}}\alpha_{{\bf x}|y} \delta_{x_y 1}=1$, it follows from the previous relation that $\sum_{{\bf x}}\alpha_{{\bf x}|y} \delta_{x_y 1}\leq \frac{1}{2}$, meaning that $x_y=0$ is a more probable event in the inputs than $x_y=1$. We can interpret this result as ``ignoring the encoding'' being the best decoding strategy Bob can implement for that particular bit, in which case it makes no sense for Alice to consider it in the encoding to begin. On the other hand, it is not hard to see that if an encoding strategy ignores the $y$-th bit, the ensuing optimal decoding is a constant function mapping the input to a constant value, which is the most frequent for $x_y$. By identifying constant decoding functions with parts of the input $\bf x$ that are ignored by the optimal strategy, we are left with only two possible decoding maps for those bits that are taken into account. This two maps are actually equivalent since it is easy to check that if $D_y(\mu)=1-\mu$ is optimal for the encoding $E$, then $D_y(\mu)=\mu$ is optimal for the encoding $E'=E\circ \neg_y$, where $\neg_y: S\mapsto S$ is the function flipping the $y$-th bit of a given string in $S$. It follows then that the optimal decoding strategy can always be chosen to be the identity map.
    
    Now we can move on to discuss the optimal encoding function. We have seen that whenever $D_y$ is a constant function the optimal encoding ignores the $y$-th bit, so we can focus here on the case where the identity map is the optimal decoding function. In that case, it follows from the second statement in Lemma \ref{EXHAUSTIVE SEARCH LEMMA} that the optimal encoding strategy should satisfy
    \begin{equation}
        \FC=\sum_{\mu, y, {\bf x}\in S_\mu} \alpha_{{\bf x} y} \delta_{x_y \mu}=\sum_{\mu, {\bf x}\in S_\mu} \alpha_{\bf x} \sum_y r_{y|{\bf x}} \delta_{x_y \mu}, \label{optenc}\
    \end{equation}
    where $r_{y|\mathbf{x}} := \alpha_{\mathbf{x} y}/\alpha_{\mathbf{x}}$, with $\alpha_{\mathbf{x}} := \sum_y \alpha_{\mathbf{x} y}$. We can think of this condition as determining which value ($0$ or $1$) has the greater weight in the string $\bf x$, that is, the optimal strategy corresponds to a {\it weighted majority} encoding.    
\end{proof}

Note that both Eqs.~\eqref{optdec} and \eqref{optenc} are easily obtained by application of statements (a) and (b) of Lemma \ref{EXHAUSTIVE SEARCH LEMMA}, which may raise the question of why this solution does not extend straightforwardly to the general case, where Bob can output more than two possible values. The key feature of the argument above is, as we have shown, that the reduced number of possible decoding functions allows us to state that for every bit there are only two options: either the bit is ignored by the strategy, or it is decoded using the identity map. Whenever $d > 2$ or $m > 2$ this is no longer true, as there exist decoding functions which are neither constant nor permutations.

\subsection{Quantum} \label{subsection quantum}

Having found a procedure to determine the optimal classical strategies, we can now move on to explore the optimal quantum strategies and the cases in which these can provide an advantage. As explained at the end of Sec.~\ref{see-saw method}, for a fixed decoding strategy $\{M^b_{y}\}$ the optimal encoding of input $\bf x$ is determined by the largest eigenvalue of the operators $\sum_{{\bf x} y} \alpha_{{\bf x}y} M^{x_y}_y$. Although simplified, this problem is still hard, since  search of the optimal measurements $\{M^b_y\}$ is to be carried over the set of all possible measurement operators. The following lemma, which applies to any RAC with two outcomes, greatly simplifies this search, by restricting it to the subclass of projective measurements.

\begin{lemma}{\label{proj_meas} }
    The value $\FC_Q$ of the $n^2\mapsto 1$ $b$-RAC can always be reached by a decoding strategy consisting only of projective measurements.
\end{lemma}
\begin{proof}
Let $\{M^0_k, M^1_k\}$ be the measurement operators over the Hilbert space $\HC$ describing Bob's decoding map for the $k$-th bit. It follows from the completeness relation $M^0_k+M^1_k=\mathbb{1}$ that the objective function in SDP \eqref{SDP} can be written as
\begin{equation}
    \begin{split}
        {\FC}_k&=\tr\left[\varrho_{k, 1} + (\varrho_{k, 0}-\varrho_{k, 1})M^0_k\right] \\
        &=\tr\left[\varrho_{k, 0} - (\varrho_{k, 0}-\varrho_{k, 1})M^1_k\right] \\
        &=\frac{1}{2} + \frac{1}{2} \tr (\varrho_{k, 0}-\varrho_{k, 1})M_k, \label{Fk}
    \end{split}
\end{equation}
where in the last line we have taken the average of the expressions in the first two, and have written $M_k = M_k^0 - M_k^1$. Hermiticity of $(\varrho_{k, 0}-\varrho_{k, 1})$ implies that we can find orthogonal subspaces ${\HC}_+$ and ${\HC}_-$, such that $\HC=\HC_+\oplus\HC_-$, spanned by its eigenvectors associated with non-negative and negative eigenvalues, respectively. If $X_\pm$ denotes the projectors onto these subspaces, it is then apparent that the optimal value for ${\FC}_k$ is attained for $M_k=X_+-X_-$, from which it follows that $M_k^0$ and $M_k^1$ can be chosen to be projectors.
\end{proof}
It is worth noting that either $M_k^0$ or $M_k^1$, in the proof above, could equal the identity operator. In such a case, one of the measurement operators would be a projector over the entire Hilbert space $\HC$, which corresponds to a decoding strategy in which Bob always guesses $0$ (or $1$) for the $k$-th bit regardless of Alice's encoding, i.e., a constant decoding map. As already seen in the discussion of optimal classical strategies, if a constant decoding function is optimal for a given bit, we can find an optimal encoding that ignores that bit. Indeed, assume without loss of generality that the optimal $k$-th decoding strategy requires $M^0_k=\mathbb{1}$. Because the optimal preparations $\rho_{\bf x}$ are eigenstates associated with the largest eigenvalue of the operators $\sum_{y} \alpha_{{\bf x}y} M^{x_y}_y$, which we now can write as $(1-x_k)\alpha_{{\bf x}k} + \lambda_{\max} (\sum_{y\neq k} \alpha_{{\bf x}y} M^{x_y}_y)$, the quantum value of the $b$-RAC value can be expressed as
\begin{equation}
\begin{split}
\FC_Q&=f_k + \sum_{\bf x} \lambda_{\max} (\sum_{y\neq k} \alpha_{{\bf x}y} M^{x_y}_y), \\
f_k&=\sum_{\bf x}(1-x_k)\alpha_{{\bf x}k} \label{split}
\end{split}
\end{equation}
where $f_k$ is the contribution to the value of the trivial decoding of the $k$-th bit, and the second term in $\FC_Q$ is the quantum value of $b$-RAC with input strings of $n-1$ bits. It follows then that the optimal preparations are eigenstates of the operators $\sum_{y\neq k} \alpha_{{\bf x}y} M^{x_y}_y$, meaning that the encoding strategy ignores the $k$-th bit as stated above.

The argument above is easily generalized to the case in which the optimal strategy ignores any number of bits. It should be noted here, as we did in discussing classical strategies in Lemma \ref{lem_optclass}, that the knowledge of the bits that should be ignored by the optimal strategy makes easier the evaluation of $\FC_Q$. However, such information does not seem to be available in advance, and can be obtained only by comparing the values of all possible bit-ignoring strategies, making the evaluation of $\FC_Q$ computationally hard. If we denote by $s$ a given subset of $I_n=\{0, 1, ..., n-1\}$, we can write the value associated with an encoding-decoding strategy ignoring the bits in $s$ as
\begin{equation}
    \FC_Q^s=\sum_{k\in s} f_k + F_Q^s \label{Fqs}
\end{equation}
where $f_k$ is again given by Eq.~\eqref{split} and
\begin{equation}
    F_Q^s=\max_{ \{M_y^{x_y}\} } \sum_{\bf x} \lambda_{\max} \Bigg(\sum_{y\notin s} \alpha_{{\bf x} y} M_y^{x_y}\Bigg)
\end{equation}
is the contribution from the bits that are not ignored by the encoding-decoding strategy. With this notation we can formally write the quantum value as
\begin{equation}
    \FC_Q=\max_{s} \FC_Q^s.
\end{equation}
It follows from this discussion that finding the $b$-RAC optimal value over quantum strategies reduces to solving, for all $s \subset I_n$, the optimization problem in the second term in Eq.~\eqref{Fqs}, i.e., finding the optimal strategies involving all bits of the input strings. In what follows, we explore the cases where an analytical solution to this problem is available.

\subsubsection{Qubit strategies}

Let us now consider $d=m=2$ and focus on optimal strategies involving all the bits in the input string $\bf x$. The decoding strategies in this scenario involve two-outcome measurements on qubits, which we can always write as convex combinations of rank-1 projective and trivial measurements, i.e., $M_y^{x_y}\in\{0, \mathds{1}\}$. Since trivial measurements are associated with bit-dropping strategies, which are assumed here to be suboptimal, it follows that the quantum value can be attained only with rank-1 projective measurements. By restricting the decoding strategies to rank-1 projective measurements we can express the $b$-RAC value as
\begin{equation}
\begin{split}
    {F}&=\sum_{{\bf x}} \tr \rho_{\bf x} \left (\sum_y \alpha_{{\bf x} y} M^{x_y}_y \right ) \\
    &=\frac{1}{2} + \sum_{\bf x} \tr \rho_{\bf x} \left (\sum_y \alpha_{{\bf x} y} (-1)^{x_y} {\bf m}_y\cdot\bm\sigma \right ) \\
    &\leq \frac{1}{2}+\sum_{\bf x} |\sum_y \alpha_{{\bf x} y} (-1)^{x_y} {\bf m}_y|, \label{qrrw}
\end{split}
\end{equation}
where in the second line we have expanded the projectors $M^{x_y}_y$ in the Pauli basis, $M^{x_y}_y=\frac{1}{2}\mathbb{1}+(-1)^{x_y}{\bf m}_y\cdot\bm\sigma$, ${\bf m}_y\in\mathbb{R}^3$, $|{\bf m}_y|=\frac{1}{2}$, and in the third line we used that the value of the trace is upper bounded by the largest eigenvalue of the traceless operator in the argument, an upper bound that is attained when  $\rho_{\bf x}$ is an eigenstate associated with this eigenvalue. It follows from this result that the quantum value of the $n^2\mapsto 1$ $b$-RAC is given by
\begin{eqnarray}
    F_Q &=& \frac{1}{2}+\max_{\{{\bf m}_y\}} \sum_{\bf x} \alpha_{\bf x} \,|\sum_y r_{y|\bf x} (-1)^{x_y} {\bf m}_y|, \label{bitvalue}
\end{eqnarray}
with $\alpha_{\bf x}=\sum_{y} \alpha_{{\bf x} y}$ and $r_{y|\bf x}=\alpha_{{\bf x} y} / \alpha_{\bf x}$. Note that, as suggested in Ref.~\cite{Ambainis09} for the case of unbiased RACs, the value $F_Q$ can be thought of as the (weighted) average distance traveled by a random walker in $\mathbb{R}^3$ (up to some scaling and shift). Moreover, it can be checked by direct calculation that if the vectors ${\bf m}_y$ are constrained to be parallel, $F_Q$ reduces to the optimal RAC value over classical strategies. 

While the optimization problem in Eq.~\eqref{qrrw} is in general hard to solve for arbitrary bias tensors, it greatly simplifies if we restrict ourselves to the subclass of {\it factorizable} biases. Consider a bias such that $\alpha_{{\bf x} y}=\alpha_{\bf x}r_y$, where $\sum_y r_y=\sum_{\bf x}\alpha_{\bf x}=1$ because of normalization. This means that the inputs $\bf x$ and $y$ are independent random variables. It is easy to see in this case that, in the sum over input strings in Eq.~\eqref{bitvalue}, the term associated with the string $\bf x$ has the same value as that associated with string $\tilde{\bf x}$ if the latter can be obtained from $\bf x$ by flipping all of its bits. Taking this into account we can rewrite $F_Q$ as the sum over only half of the input strings
\begin{equation}
\begin{split}
    F_Q &= \frac{1}{2}+\max_{\{{\bf m}_y\}} \sum_{\bf x} p_{\bf x} \,|\sum_y r_y (-1)^{x_y} {\bf m}_y| \\
    &= \frac{1}{2}+\max_{\{G\}} \sum_{\bf x} p_{\bf x} \sqrt{\tr G\, {\bf v}_{\bf x} {\bf v}_{\bf x}^T }, \label{gsqrt} 
\end{split}
\end{equation}
with $p_{\bf x}=\alpha_{\bf x}+\alpha_{\tilde{\bf x}}$, ${\bf v}_{\bf x}$ an $n$-dimensional tuple with components ${\bf (v}_{\bf x})_y=(-1)^{x_y}r_y$ and $G$ the Gram matrix of the measurement vectors ${\bf m}_y$, i.e., $G_{ij}=\braket{{\bf m}_i, {\bf m}_j}$. We can now think of the sum over $\bf x$ in Eq.~\eqref{gsqrt} as a scalar product between two $2^{n-1}$-dimensional tuples, one with components $p_{\bf x}$ and the other with components $\sqrt{\tr G\, {\bf v}_{\bf x} {\bf v}_{\bf x}^T}$. Using the Cauchy-Schwarz inequality we can upper bound $F_Q$ by
\begin{equation}
    \begin{split}
        F_Q &\leq \frac{1}{2}+\sqrt{\sum_{\bf x} p_{\bf x}^2}\, \max_{G} \sqrt{\tr G\sum_{\bf x}{\bf v}_{\bf x}{\bf v}_{\bf x}^T } \\
        &=\frac{1}{2}+\sqrt{2^{n-3}}\sqrt{\sum_{\bf x}p_{\bf x}^2 }\sqrt{\sum_{y} r_y^2},
    \end{split} \label{upb}
\end{equation}
where in the last line we used that $(\sum_{\bf x}{\bf v}_{\bf x}{\bf v}_{\bf x}^T)_{ij}=2^{n-1}r^2_i\delta_{ij}$, as can be checked by direct calculation. It is easy to check that this expression reduces, for $\alpha_{{\bf x}y}=\frac{1}{n2^n}$, to the upper bound $F_Q\leq \frac{1}{2}+(2\sqrt{n})^{-1}$ previously derived in Ref.~\cite{Ambainis09}. It is also worth remarking that the upper bound we just derived depends on $p_{\bf x}$ rather than directly on $\alpha_{\bf x}$, and this feature holds for the value associated with any quantum or classical strategy, since it is a consequence of the independence of inputs $\bf x$ and $y$.

The bound in Eq.~\eqref{upb} was obtained via the Cauchy-Schwarz inequality between two tuples with components $\alpha_{\bf x}$ and  $\sqrt{\tr G\, {\bf v}_{\bf x} {\bf v}_{\bf x}^T }$, respectively, which are all real and non-negative. Thus, in order to saturate the bound following $2^{n-1}$ conditions (one per input string) must be satisfied
\begin{equation}
    \frac{p_{\bf x}}{\sqrt{\tr G{\bf v}_{\bf x}{\bf v}_{\bf x}^T}} = \frac{1}{\sqrt{2^{n-3}}}\sqrt{\frac{\sum_{{\bf x}'} p_{\bf x'}^2}{\sum_{y}r_y^2}}, \label{cond}
\end{equation}
where as before we have written $p_{\bf x}=\alpha_{\bf x}+\alpha_{\tilde{\bf x}}$. Note that the quantity on the right-hand side is a constant characterizing the particular $b$-RAC under study. Since projectivity of measurements implies the norms of the vectors ${\bf m}_y$ is maximal, $|{\bf m}_y|=\frac{1}{2}$, we can use that $\tr G{\bf v}_{\bf x}{\bf v}_{\bf x}^T=\frac{1}{4}\sum_{i j}(-1)^{x_i+x_j}r_ir_j\cos(\theta_{i j})$ to rewrite Eq.~\eqref{cond} as a condition to be satisfied by the angles $\theta_{i j}$ between vectors,
\begin{equation}
    \sum_{i<j}(-1)^{x_i+x_j} r_i r_{j}\cos(\theta_{i j})=\frac{1}{2} \left (\sum_{y} r_y^2 \right) \! \left(  \frac{2^{n-1} p^2_{\bf x}}{\sum_{\bf x'}p_{{\bf x}'}^2} -1 \right). \label{ang_cond}
\end{equation}
Now fix a pair of indices $(i, j)$, with $j\neq i$, and define $S_{ij}$ as the subset of input strings satisfying $x_i=x_j$. Then, summing over the strings in $S_{ij}$ in Eq.~\eqref{ang_cond} we arrive at
\begin{equation}
    \cos(\theta_{i j})=\frac{1}{2r_ir_j}\left(\frac{\sum_y r_y^2}{\sum_{\bf x'} 
 p^2_{\bf x'}}\right) \left(\sum_{{\bf x}\in S_{ij}} p^2_{\bf x}-\sum_{{\bf x}\notin S_{ij}} p^2_{\bf x} \right). \label{gcond}
\end{equation}

If instead of the vectors $\{m_y\}$ we parametrize the set of optimal measurements by the cosine of the angles between them, then Eq.~\eqref{gcond} shows how to construct the optimal decoding strategy from the biasing parameters, when the upper bound \eqref{upb} is attained. On the other hand, we can use the relation to quickly discard the possibility of the bound being attained: If for a given $b$-RAC the cosines generated by Eq.~\eqref{gcond} are inconsistent, then the upper bound cannot be attained. Consistency here turns out to be equivalent to (i) $\cos(\theta_{ij})\in [-1, 1]$ and (ii) the matrix with elements 
\begin{equation}
    \Tilde{G}_{ij}=\begin{cases}
                    \frac{1}{4} & \text{if}\; i=j \\
                    \frac{1}{4} \cos{(\theta_{ij})} & \text{if}\; i\neq j 
                   \end{cases} \label{Gtilde}
\end{equation}
being positive semidefinite.

It turns out that in the $2^2\mapsto 1$ scenario whenever quantum strategies can provide some advantage over their classical counterpart, this upper bound is actually attainable, as shown in the following lemma.
\begin{lemma}\label{lemma_2_bit}
    The quantum value of a $2^2\mapsto 1$ $b$-RAC with biasing strategy $\alpha_{{\bf x} y}=\alpha_{{\bf x}}r_y$ is
    \begin{equation}
        F_Q=\max\left\{F_C,\; \frac{1}{2}+\frac{1}{\sqrt{2}} \sqrt{\sum_{\bf x} p_{\bf x}^2} \sqrt{\sum_{y} r_y^2}\right\},
    \end{equation}
    where $F_C=\frac{1}{2}+\max\{ \frac{1}{2}p_{00}+\frac{1}{2}p_{01}(r_0-r_1),\, \frac{1}{2}p_{00}(r_0-r_1)+\frac{1}{2}p_{01} \}$, with $p_{\bf x}=\alpha_{\bf x}+\alpha_{\tilde{\bf x}}$, is the optimal performance over classical strategies.
\end{lemma}
\begin{proof}
    Since $n=2$ the sum over $y$ in Eq.~\eqref{gsqrt} has only two terms. This means, since vectors $\bm m_y$ have maximal norm, that there is only one parameter to optimize over, which is the angle $\theta$ between $\bm m_0$ and $\bm m_1$. Looking for critical values, we find that the $b$-RAC functional can attain a maximal value only for $\theta$ satisfying $\sin(\theta)=0$ or
    \begin{equation}
    \cos(\theta)=\left( \frac{r_0^2+r_1^2}{2 r_0 r_1}\right) \left(\frac{p_{00}^2-p_{01}^2}{p_{00}^2+p_{01}^2} \right), \label{angles 2 bits into qubit}
    \end{equation}
    which is just the condition in Eq.~\eqref{gcond}, implying that $F_Q$ is then given by Eq.~\eqref{upb}. On the other hand, if $\theta$ equals $0$ or $\pi$, then $\FC_Q$ reads $\frac{1}{2}+\frac{1}{2}p_{00}+\frac{1}{2}p_{01}(r_0-r_1)$ and $\frac{1}{2}+\frac{1}{2}p_{00}(r_0-r_1)$, respectively. These two quantities can be checked to correspond to the two options present in the expression for the classical value.
\end{proof}
Lemma~\ref{lemma_2_bit} shows not only that the upper bound given in Eq.~\eqref{upb} can be attained, but that it will be attained with any bias for which quantum strategies provide any advantage over classical strategies, thus providing a complete solution to the $2^2\mapsto 1$ $b$-RAC value problem. Moreover, it follows from the proof that attaining the upper bound in Eq.~\eqref{upb} self-tests the angle $\theta$ between the Bloch vectors defining the optimal measurements. Unfortunately such a simple result does not hold for a larger number of bits, with solutions becoming more complex already in the case of $3$-bit input strings. For this particular case, nonetheless, the upper bound in Eq.~\eqref{upb} still can be attained for many biasing tensors, as we show in the following lemma, which extends the proof of Lemma~\ref{lemma_2_bit}.
\begin{lemma}\label{FQn=3}
    The quantum value $F_Q$ of a $3^2\mapsto 1$ $b$-RAC is given by
    \begin{equation}
        F_Q=\frac{1}{2} + \sqrt{\sum_{\bf x} p_{\bf x}^2} \sqrt{\sum_{y} r_y^2},
    \end{equation}
    whenever there exists an optimal decoding strategy for which the Bloch vectors $\{{\bf m}_{y}\}$ are linearly independent.
\end{lemma}
\begin{proof}
    If the vectors are linearly independent, then the ensuing Gram matrix $G$ is full rank. As we did in Eq.~\eqref{gsqrt}, we can write the optimal performance in terms of the elements of $G_{i, j} = \braket{{\bf m}_i, {\bf m}_j} = \frac{1}{4} \cos{(\theta_{ij})}$, and in doing so parametrize $F_Q$ by the angles $\theta_{ij}$, which are all independent. As a consequence, a critical point will satisfy $\partial_{\theta_{ij}} F_Q=\partial_{G_{ij}}F_Q \sin{\theta_{ij}}=0$, which implies $\partial_{G_{ij}}F_Q=0$ because of the independence of the measurement operators. A direct calculation shows then that a critical point $\{\theta_{ij}\}$ should be such that the following relation is satisfied $\forall \, i, j$:
    \begin{equation}
        \sum_{\bf x}\frac{p_{\bf x}}{\sqrt{\tr G{\bf v}_{\bf x} {\bf v}_{\bf x}^T}}(-1)^{x_i+x_j}=0.
    \end{equation}
    This equation can be rewritten as a matrix equation $A\cdot b=0$ with
    \begin{equation}
        b_{\bf x}=\frac{p_{\bf x}}{\sqrt{\tr G{\bf v}_{\bf x} {\bf v}_{\bf x}^T}},
    \end{equation}
    and $A\in\mathbb{R}^{3\times 4}$ a matrix with entries $A_{(i, j), \bf x}=(-1)^{x_i+x_j}$. It is straightforward to check that the null space of $A$ is ${\rm Null}(A)={\rm span}\{(1, 1, 1, 1)\}$, implying that in a critical point the quotient $\frac{p_{\bf x}}{\sqrt{\tr G{\bf v}_{\bf x} {\bf v}_{\bf x}^T}}$ takes the same value for all strings $\bf x$, and that this value is $\sqrt{\frac{\sum_{\bf x} p_{\bf x}^2}{2^{n-3}\sum_y r_y^2}}$, thus proving that the only critical point in $F_Q$ satisfies condition \eqref{cond} and $F_Q$ is given by the upper bound in Eq.~\eqref{upb}.
\end{proof}
The above lemma shows that the cosines built up from the biasing parameters via the condition in Eq.~\eqref{gcond} will indeed provide the optimal measurements if, in addition to being consistent, the matrix $\tilde G$ of Eq.~\eqref{Gtilde} is full rank. Moreover, it is clear that attaining the upper bound in Eq.~\eqref{upb} self-tests the angles satisfying Eq.~\eqref{gcond} since the optimal performance is achieved only by satisfying these relations. We can give a geometric interpretation to these conditions by noting that positive semidefiniteness of $\tilde{G}$ is ensured if its determinant is non negative. We can write this last condition as $1-\sum_{i<j}\cos^2({\theta_{i j}})+2\prod_{i< j}\cos(\theta_{i j})\geq 0$, which is the equation of an ``inflated tetrahedron'' centered at the origin. Note that this origin corresponds to $\cos(\theta_{i j})=0\;\forall i\neq j$, i.e., the angles associated with mutually unbiased measurements, which are in turn the optimal decoding strategy for the case of unbiased input strings $\bf x$, as can be easily checked in Eq.~\eqref{gcond}. Therefore, we have that the upper bound in Eq.~\eqref{upb} becomes the quantum value whenever the optimal measurements, as described by the three cosines $\{ \cos(\theta_{ij})=4\braket{{\bf m}_i, {\bf m}_j} \}$, live inside the inflated tetrahedron. 

The geometrical picture introduced above turns out to be very helpful in understanding how different measurements become optimal as the bias tensor $\alpha_{{\bf x} y}$ departs from $\alpha_{{\bf x} y}=\frac{1}{3 2^3}$, which is the unbiased case. In what follows we study the solution when the bias tensor corresponds to one of the built-in biasing functions described in Section~\ref{S:package}. A similar analysis for other of these built-in biases can be found in Appendix~\ref{Ap:sec4}.

\subsubsection{The \texorpdfstring{\protect\Verb+X\_ONE+}{} and \texorpdfstring{\protect\Verb+Y\_ALL+}{} bias family}

We can have a better understanding about how different decoding strategies become optimal for different biases by analyzing a few examples. Consider the case of an input string bias of the form
\begin{equation}
    \alpha_{\bf x}=\begin{cases}
        w &\; \text{if} \;  {\bf x}=000 \\
        \frac{1-w}{7} &\; \text{otherwise}
        \end{cases}, \quad 0\leq w\leq 1 \label{x_one_bias}
\end{equation}
in combination with an arbitrary distribution $\{r_y\}$ for the requested bit. A numerical analysis suggests that for all members of the family the optimal quantum strategy involves all three bits in the input string, in which case the optimal performance is given by Eq.~\eqref{bitvalue}. In what follows then we will restrict our attention to those strategies.

The optimal decoding strategy for $w$ in the vicinity of $\frac{1}{8}$ is expected to be determined by Eq.~\eqref{gcond}. We find, for a bias of this form, that the upper bound in Eq.~\eqref{upb} is attainable if the biasing coefficients are such that the three cosines
\begin{equation}
    \cos[\theta_{i j}(w)]=\frac{\sum_y r_y^2}{2r_ir_j}h(w), \label{cos_x_one}
\end{equation}
where
\begin{equation}
    h(w)=\frac{32w^2+20w-3}{48w^2-12w+13}
\end{equation}
is an increasing function for $w\in[0, 1]$, and ${h(\frac{1}{8})=0}$, are found to be consistent. Note that from Eq.~\eqref{cos_x_one} it follows that the optimal strategy $(\cos[\theta_{01}(w)],~\cos[\theta_{02}(w)],$ $\cos[\theta_{12}(w)])\in\mathbb{R}^3$ continuously departs from the center of the tetrahedron as $w$ increases, in a direction that is specified by both the sign of $h(w)$ and the fixed but otherwise arbitrary choice of the $\{r_y\}$ bias. In what follows we will focus on the case of $h(w)\geq 0$, i.e., $w\in[1/8, 1]$, since it is enough to understand the behavior of the solutions in the entire interval.

For uniform $\{r_y\}$ we see that the cosines in Eq.~\eqref{cos_x_one} are the same $\forall i, j$ and consistent for $\frac{1}{8}\leq w\leq 5/12$, increasing from $0$ to $1$ as $w\in[1/8, 1]$ increases. Within this region, quantum strategies have an advantage over the classical ones except for $w=5/12$, in which case the two values coincide, $F_Q=F_C$. The optimal decoding strategy moves from the center of the inflated tetrahedron described above to its $(1, 1, 1)$ vertex as $w$ increases within $[\frac{1}{8}, \frac{5}{12}]$. For biasing parameter $w>\frac{5}{12}$ the cosines determined by Eq.~\eqref{cos_x_one} stop being consistent and the optimal strategy is expected to remain in the vertex of the tetrahedron, as follows from Lemma \ref{FQn=3} and the symmetry of the solution. The optimal performance resulting from this analysis is shown in FIG.~\ref{fig:X_ONE}, in comparison with the numerical value obtained with the see-saw procedure described in Sec.~\ref{see-saw method}.

Now let $\{r_y\}$ be biased such that $r_0>r_1=r_2$. In this situation, the condition in Eq.~\eqref{cos_x_one} implies $\cos[\theta_{01}(w)]=\cos[\theta_{02}(w)]<\cos[\theta_{12}(w)]$. As in the previous case the point in $\mathbb{R}^3$ representing the optimal strategy moves away from the coordinate origin as $w$ increases, but now towards a point located on the edge of the tetrahedron connecting the $(1, 1, 1)$ and $(-1, -1, 1)$ vertices, and reaching it for a given value $w_c$ of the biasing parameter satisfying $h(w_c)=\frac{2r_1^2}{r_0^2+2r_1^2}$. For larger values of $w$, $F_Q$ is no longer given by Eq.~\eqref{upb}, but numerics suggest that the optimal solution remains on the edge of the tetrahedron, i.e., $\cos[\theta_{12}(w)]=1$ for $w>w_c$, moving towards the $(1, 1, 1)$ vertex as $w$ increases. We can find the ensuing value $F_Q$ by imposing ${\bf m}_1={\bf m}_2$ in the first line of Eq.~\eqref{gsqrt}, which leads to
\begin{widetext}
\begin{equation}
\begin{split}
    F_Q(w>w_c)&=\frac{1}{2}+ \frac{r_0}{2}(p_{010} + p_{001}) + \max_{\{ {\bf m_0}, {\bf m_1} \}} \left\{ p_{000}|r_0 {\bf m}_0 + 2r_1 {\bf m}_1| + p_{100}|r_0 {\bf m}_0 - 2r_1 {\bf m}_1|\right\} \\
    &\leq \frac{1}{2}+\frac{r_0}{2}(p_{010} + p_{001}) +\frac{1}{\sqrt{2}}\sqrt{r_0^2+4r_1^2}\sqrt{p_{000}^2+p_{100}^2}
\end{split} 
\end{equation}
\end{widetext}
with the upper bound in the second line being attained for $\cos[\theta_{01}]=\cos[\theta_{02}]= c_\text{opt}$, where
\begin{equation}
    c_\text{opt}=\left(\frac{r_0^2+4r_1^2}{4r_0r_1}\right)\left(\frac{p_{000}^2-p_{100}^2}{p_{000}^2+p_{100}^2}\right)
    \label{condition for cosines}
\end{equation}
follows from Eq.~\eqref{angles 2 bits into qubit}. As already discussed, $F_Q$ coincides with this upper bound when the cosines produced by Eq.~\eqref{condition for cosines} as long as $-1< c_\text{opt} <1$, becoming equal to the optimal classical value otherwise. The piecewise-defined optimal performance,
\begin{equation}
    F_Q=\begin{cases}
        \frac{1}{2}+\sqrt{\sum_{\bf x}p_{\bf x}^2 }\sqrt{\sum_{y} r_y^2} \;\; &\text{if} \;\; w\leq w_c \\
        \frac{1}{2}+\frac{r_0}{2}(p_{010} + p_{001}) \\
        \:\quad + \frac{1}{\sqrt{2}} \sqrt{r_0^2+4r_1^2} \sqrt{p_{000}^2+p_{100}^2} \;\; &\text{if} \;\; w\geq w_c,
    \end{cases}
    \vspace{2mm}
\end{equation}
results from this analysis and is depicted in FIG.~\ref{fig:X_ONE_Y_ONE}, together with the results obtained from the numerical package and the upper bound value in Eq.~\eqref{upb}, for the case $r_0=\frac{1}{2}$ and $r_1=r_2=\frac{1}{4}$. 

\begin{figure}[t]
    \centering
    \hfill\includegraphics[width = 8.6 cm]{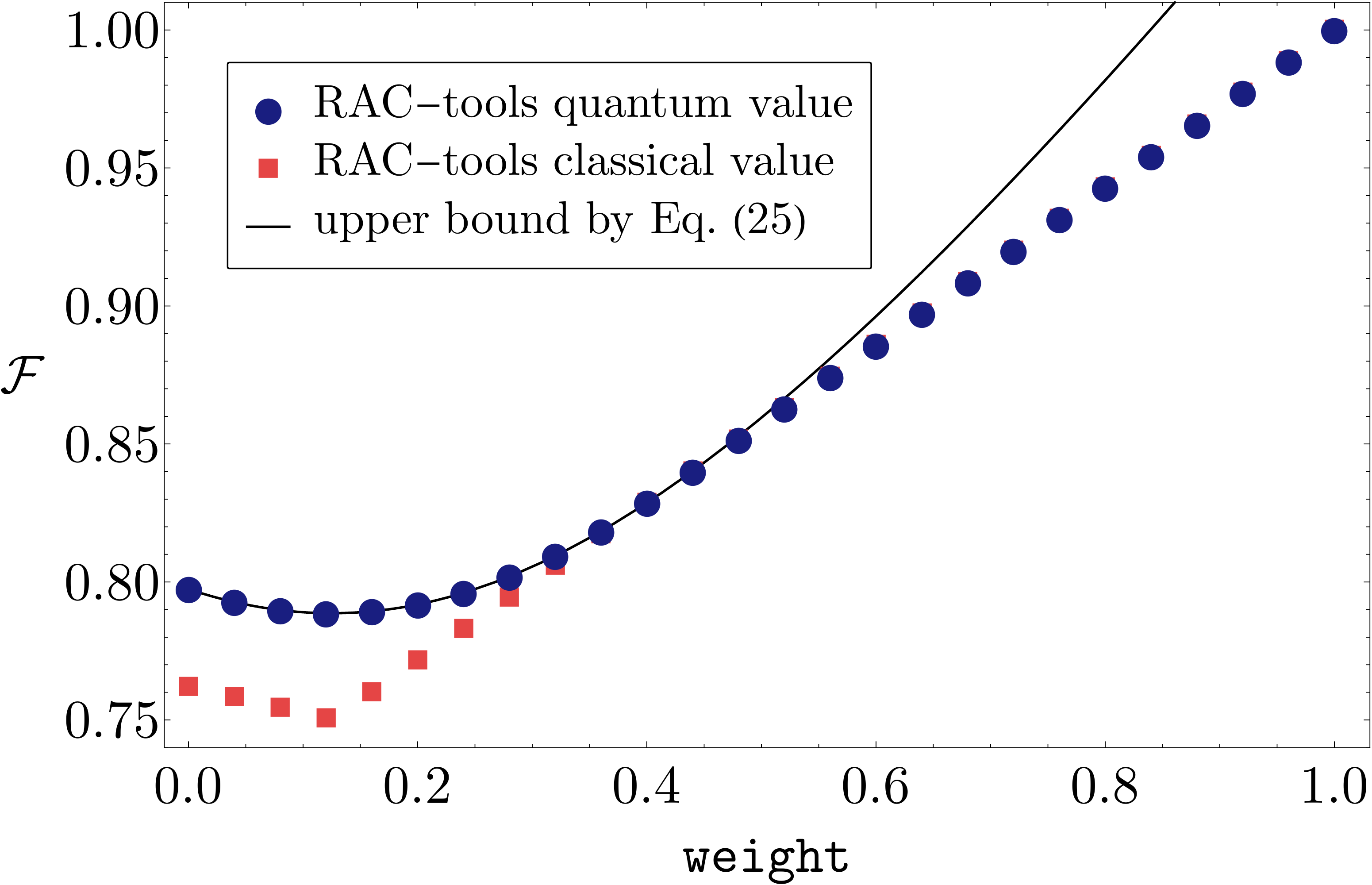}\\
    \vspace{2.5 mm}
    \hfill\includegraphics[width = 8.4 cm]{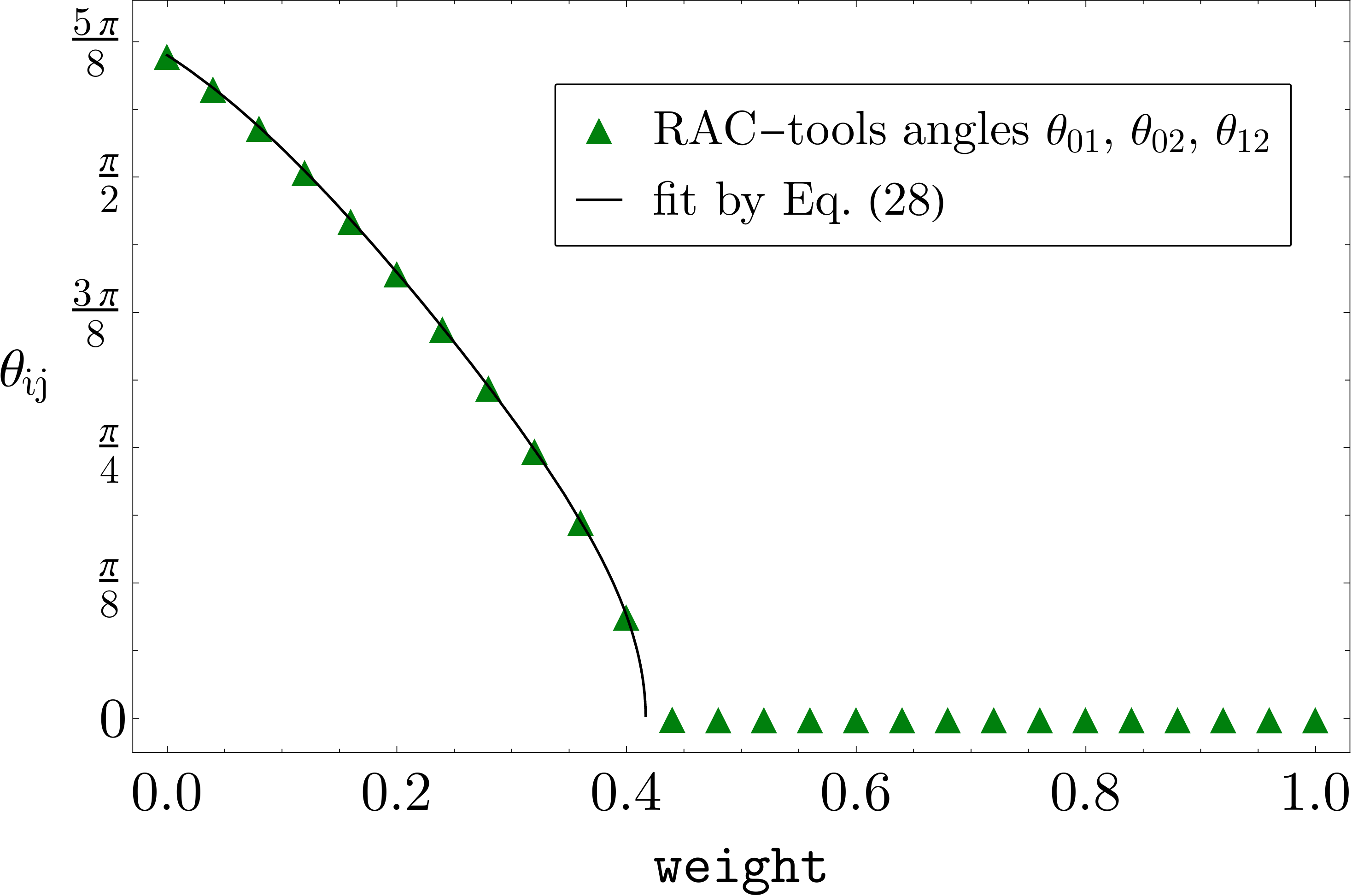}
    \caption{Top: Optimal performance of the $3^2\mapsto 1$ RAC, with \protect\Verb+X\_ONE+ bias, over classical (red squares) and quantum (blue dots) encoding-decoding strategies. For $w\in [1/8, 5/12)$ the quantum value is given by the upper bound in Eq.~\eqref{upb} (solid line) and is strictly larger than its classical counterpart. For larger values of the biasing parameter the two values coincide. Bottom: Angles $\theta_{ij}$ between the Bloch vectors defining the measurement operators in the optimal decoding strategy. As expected from the symmetry of the bias tensor, the three angles coincide and have positive values in the quantum advantage region $w\in [1/8, 5/12)$, vanishing for $w\geq 5/12$. }\label{fig:X_ONE}
\end{figure}

\begin{figure}[t]
    \centering
    \hfill\includegraphics[width = 8.6 cm]{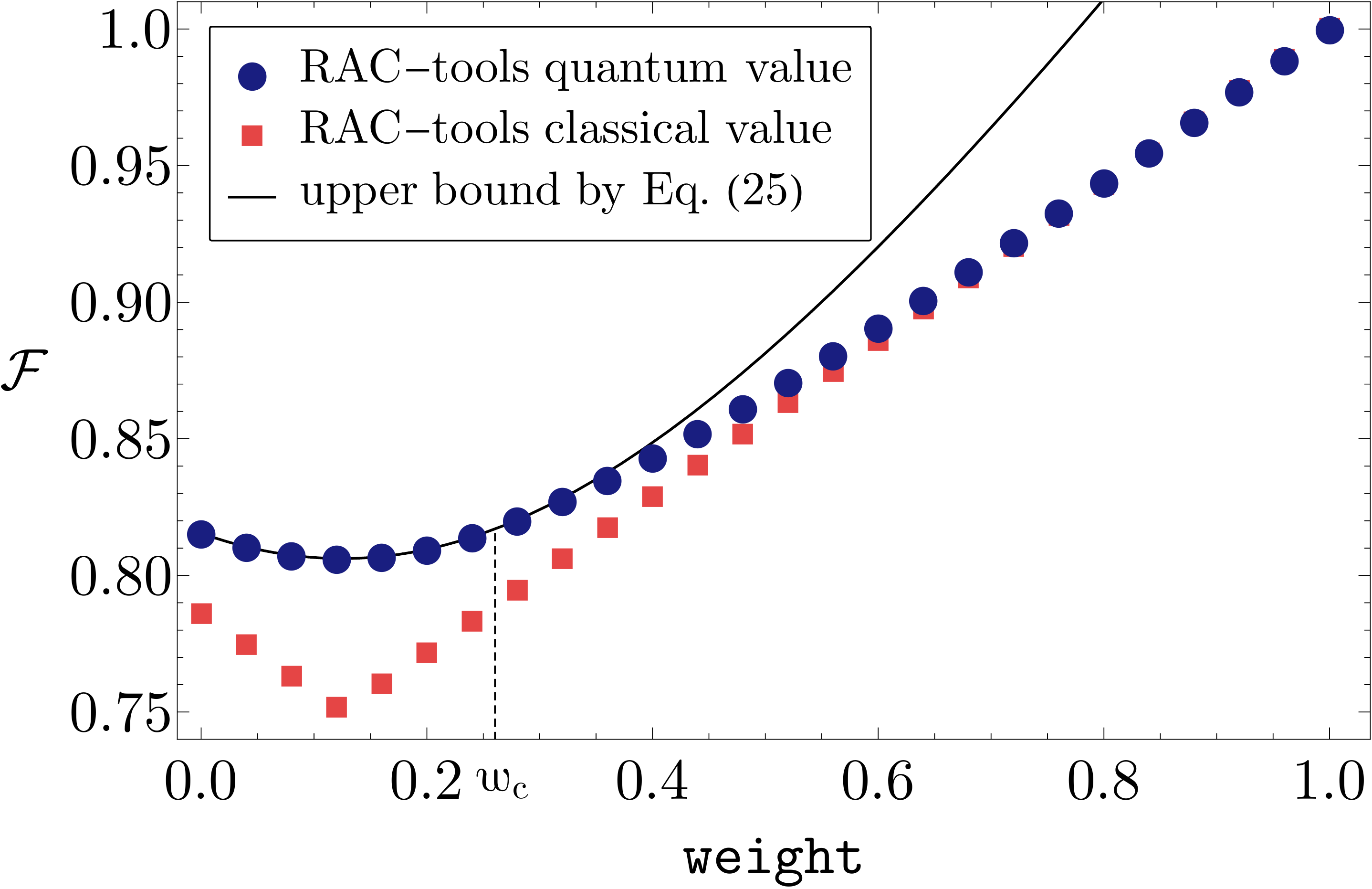}\\
    \vspace{2.5 mm}
    \hfill\includegraphics[width = 8.3 cm]{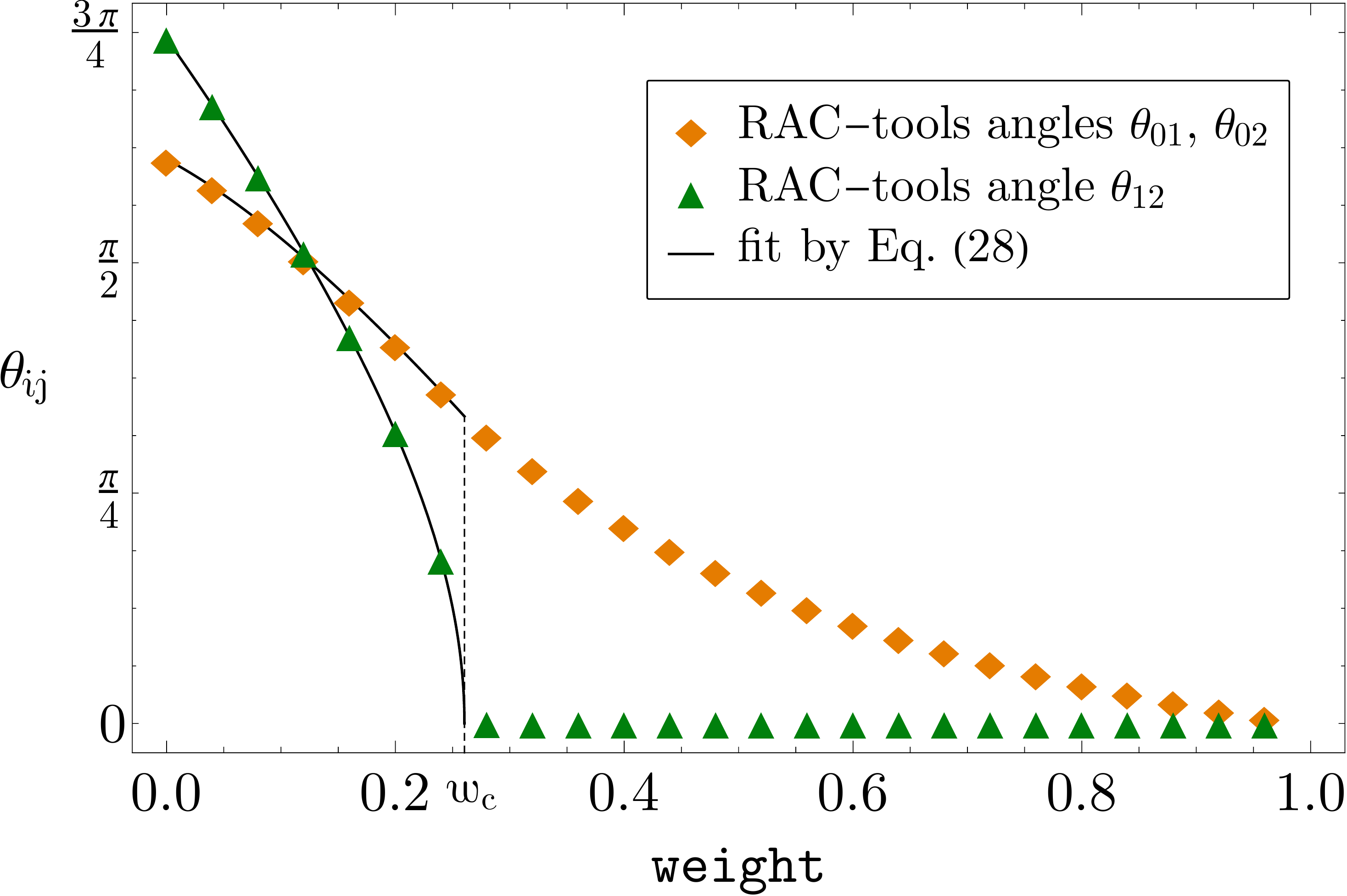}
    \caption{ Top: Optimal performance of the $3^2\mapsto 1$ RAC with \protect\Verb+X\_ONE+ bias and $r_0=0.5$, $r_1=r_2=0.25$. As a result of the bias in the requested bit, one of the conditions in Eq.~\eqref{cos_x_one} is saturated before the other two, giving rise to a region in which the optimal quantum performance (blue dots) is strictly larger than the optimal classical value (red squares) but nonetheless strictly smaller than the upper bound in Eq~\eqref{upb} (solid line). This region is found to be $w> \frac{1}{12}(7\sqrt{3} - 9)$. Bottom: Angles $\theta_{ij}$ parametrizing the optimal decoding strategy. Because of the asymmetry in the bias, one of the angles, $\theta_{12}$ (orange diamonds), is different from the other two (green triangles) and decreases faster as a function of the biasing parameter $w$, vanishing exactly at $w=\frac{1}{12}(7\sqrt{3} - 9)$.} \label{fig:X_ONE_Y_ONE}
\end{figure}

For any other bias in $\{r_y\}$, the solution will have a similar behavior as a function of $w$, departing from the origin as the parameter increases but reaching the boundary of the inflated tetrahedron somewhere over one of the curved faces at a given critical value $w_c$ of the biasing parameter. For $w>w_c$ the optimal measurements are no longer independent and the solution remains on the boundary, moving towards the vertex as $w$ approaches $1$ (in the limit $w\mapsto 1$ only the contribution from $\mathbf{x} = 000$ string is relevant, and therefore all the angles between measurements tend to zero). Finally, note that if we choose an analogous bias in which the weight $w$ in Eq.~\eqref{x_one_bias} is assigned to a string different from $000$, we will reach a similar conclusion except that some of the cosines in Eq.~\eqref{cos_x_one} might become negative, and for $w\mapsto 1$ the optimal solution might converge to a different vertex of the tetrahedron.


\section{Analytical results for the \texorpdfstring{$2^{\MakeLowercase{d}}$}{} \texorpdfstring{$\mapsto 1$}{} RAC \label{S:2d}} \label{S:2d}

After analyzing the quantum value of different $b$-RACs in the qubit setting, here we explore the strategies attaining the quantum value of the $2^d\mapsto 1$ $b$-RAC. In this scenario, if we consider a factorizable bias tensor $\alpha_{{\bf x} y}=\alpha_{\bf x} r_y$ and a quantum realization $\{\rho_{x_0x_1}, M_{0}^{x_0}, M_{1}^{x_1}\}$, the $b$-RAC value is given by
\begin{equation}
    \FC=\sum_{x_0x_1} \alpha_{x_0x_1} \tr[\rho_{x_0x_1}(r_0 M_0^{x_0} + r_1 M_1^{x_1})], \label{f2d}
\end{equation}
with $x_i\in\{0, ..., d-1\}$, and where $\rho_{x_0x_1}$ and $M_y^{x_y}$ are operators over a $d$-dimensional Hilbert space. 

For the particular case of $\alpha_{x_0x_1}=\frac{1}{d^2}$ and $r_y=\frac{1}{2}$ it is known that the quantum value can be attained only with rank-1 projective measurement operators \cite{Farkas19}. On the other hand, the results produced by the RAC-tools package for $2^d\mapsto 1$ $b$-RACs with $d\leq 6$ (see Appendix~\ref{Ap:D}) suggest that optimizing over projective measurements might already be enough to find the quantum value. In the following, we will then restrict ourselves to finding the optimal quantum strategies using projective measurements. As we will see in the following lemma, it is possible in this scenario to derive an upper bound analogous to Eq.~\eqref{upb} found for the $n^2\mapsto 1$ $b$-RAC.

\begin{lemma}\label{lemma_FQd}
    The optimal value over projective measurements, $\FC_P$, of the $2^d\mapsto 1$ $b$-RAC defined by the bias tensor $\alpha_{{\bf x}y}=\alpha_{\bf x}r_y$ satisfies
    \begin{equation}
        \FC_P\leq \frac{1}{2}+\frac{1}{2}\sqrt{d^2-4d(d-1)r_0r_1}\sqrt{\sum_{x_0x_1} \alpha_{x_0x_1}^2}. \label{upb-d}
    \end{equation}
\end{lemma}
\begin{proof}
    We begin by noting that $r_0 M_0^{x_0}+r_1 M_1^{x_1}$ is positive semidefinite, and as a consequence the value of the functional in Eq.~\eqref{f2d} is upper bounded by
    \begin{equation}
        \FC_P\leq \sum_{x_0 x_1} \alpha_{x_0 x_1} \lambda_{x_0x_1},
    \end{equation}
    where $\lambda_{x_0x_1}$ denotes the largest eigenvalue of $r_0 M_0^{x_0}+r_1 M_1^{x_1}$. Because the measurements are assumed to be projective, by Jordan's Lemma there is a basis in which the operators $M_0^{x_0}$ and $M_1^{x_1}$ are jointly block-diagonal, with blocks of dimension $1$ or $2$. The restriction of these projectors to the $k$-th Jordan block, $P_k$ and $Q_k$, respectively, are rank-1 projectors regardless of the block dimension. The angle between the pure states on which they project, defined by $\cos^2{(\theta_k)}=\tr P_k Q_k $, is one of the {\it principal angles} between $M_0^{x_0}$ and $M_1^{x_1}$. The principal angle defines the coefficients of $P_k$ and $Q_k$ when the block is two-dimensional, which are given by
    {\allowdisplaybreaks
    \begin{subequations}
    \begin{align}
        P_k&=\frac{1}{2}\begin{bmatrix}
                1+\cos{(\theta_k)} & \sin{(\theta_k)} \\
                \sin{(\theta_k)} & 1-\cos{(\theta_k)}
            \end{bmatrix}, \\
        Q_k&=\frac{1}{2}\begin{bmatrix}
                1+\cos{(\theta_k)} & -\sin{(\theta_k)} \\
                -\sin{(\theta_k)} & 1-\cos{(\theta_k)}
            \end{bmatrix}.
    \end{align}
    \end{subequations}    
    }
    It follows then that we can write
    \begin{equation}
        r_0 M_0^{x_0}+r_1 M_1^{x_1}=\sum_{k} r_0 P_k + r_1 Q_k, 
    \end{equation}
    implying that $\lambda_{x_0x_1}$ is the largest eigenvalues of one of the blocks $r_0 P_k + r_1 Q_k$. If the corresponding block is two-dimensional, a direct calculation shows it is given by
    \begin{equation}
    \begin{split}
        \lambda_{x_0x_1} &= \frac{1}{2}\left[1 + \max_{k} \sqrt{1-4r_0r_1\sin^2(\theta_k)}\right] \\
        &\leq \frac{1}{2}\left[1+\sqrt{1 + 4r_0r_1(\tr M_0^{x_0}M_1^{x_1} -1) }\right]
    \end{split}, \label{lambda_max}
    \end{equation}
    where, in the last line, we used that $\cos^2{(\theta_k)}=\tr P_k Q_k $ and that $\tr P_k Q_k\leq\tr M_0^{x_0} M_1^{x_1} \;\forall \, k$. If the block is one-dimensional, then it is easy to see that $\lambda_{x_0x_1}\in\{0, r_0, r_1, 1\}$, in which case the upper bound in Eq.~\eqref{lambda_max} also holds. Combining this upper bound with the Cauchy-Schwarz inequality we arrive at
    \begin{equation}
        \FC_P\leq \frac{1}{2}+\frac{1}{2}\sqrt{d^2-4d(d-1)r_0r_1}\sqrt{\sum_{x_0x_1} \alpha_{x_0x_1}^2},
    \end{equation}
    where we have used the completeness relation satisfied by the measurement operators to write $\sum_{x_0x_1}\tr M_0^{x_0}M_1^{x_1}=d$.
\end{proof}

As follows from the the proof given above, attaining the upper bound of Lemma~\ref{lemma_FQd} is possible if there exist projectors $M_y^{x_y}$ such that the upper bound in Eq.~\eqref{lambda_max} is saturated and condition
\begin{equation}
    \sqrt{1+4r_0r_1(\tr M_0^{x_0}M_1^{x_1} - 1)}\propto \alpha_{x_0 x_1}, \label{cauchy-schwarz-d}
\end{equation}
is satisfied, where the proportionality constant is easily found to be
\begin{equation}
 C=\sqrt{\frac{d^2-4r_0r_1d(d-1)}{\sum_{x_0x_1}\alpha^2_{x_0x_1}}} .  
\end{equation}
Now note that the measurement operators satisfy $\tr M_y^{x_y}={\rm rank}(M_y^{x_y})$, since we are considering projective measurements. Then, squaring equation Eq.~\eqref{cauchy-schwarz-d} and summing over either $x_0$ or $x_1$ we arrive at
\begin{equation}
    {\rm rank}(M_y^{x_y})=\frac{1}{4r_0r_1}\left[ d(4r_0r_1-1)+C^2\sum_{x_{1-y}} \alpha^2_{x_0x_1} \right]. \label{rank}
\end{equation}
It is then easy to check that for $\alpha_{x_0x_1} = \frac{1}{d^2}$ Eq.~\eqref{rank} implies that the upper bound in Eq.~\eqref{upb-d} can be attained only with rank-1 measurement operators. In fact, by virtue of Eq.~\eqref{cauchy-schwarz-d} we have that the optimal measurements are mutually unbiased, thus recovering the solution reported in \cite{Farkas19} for the case $r_0=r_1=\frac{1}{2}$, and extending it to arbitrary biases on Bob's input. FIG.~\ref{fig:Y_ONE_d} shows the agreement of the optimal value provided by the numerical package and the upper bound in Lemma~\ref{lemma_FQd} for the cases $d=3$, $d=4$ and $d=5$.
\begin{figure}[b!]
    \centering
    \hfill\includegraphics[width = 8.45 cm]{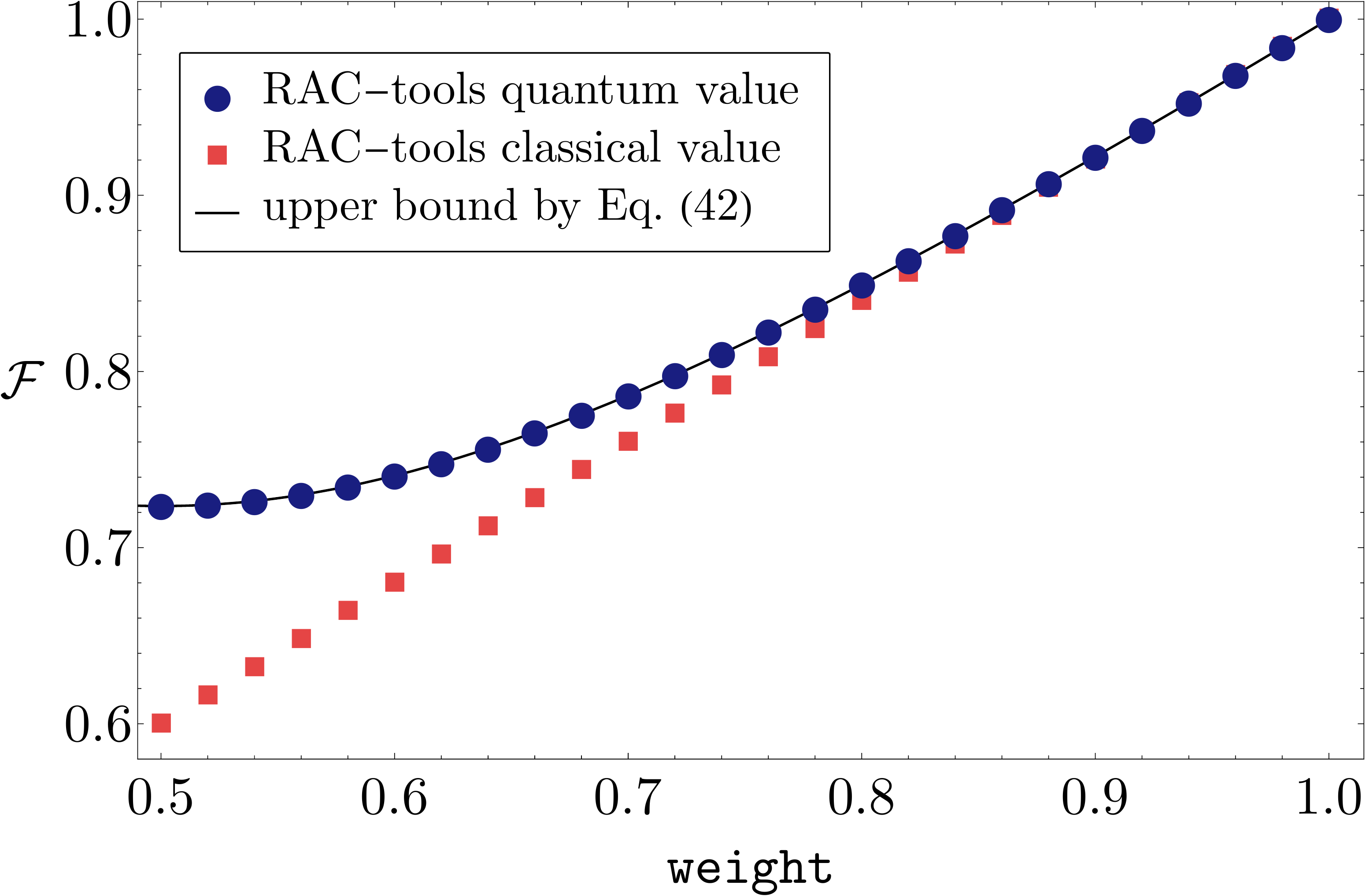}\\
    \vspace{-1 mm}
    \hfill\includegraphics[width = 8.65 cm]{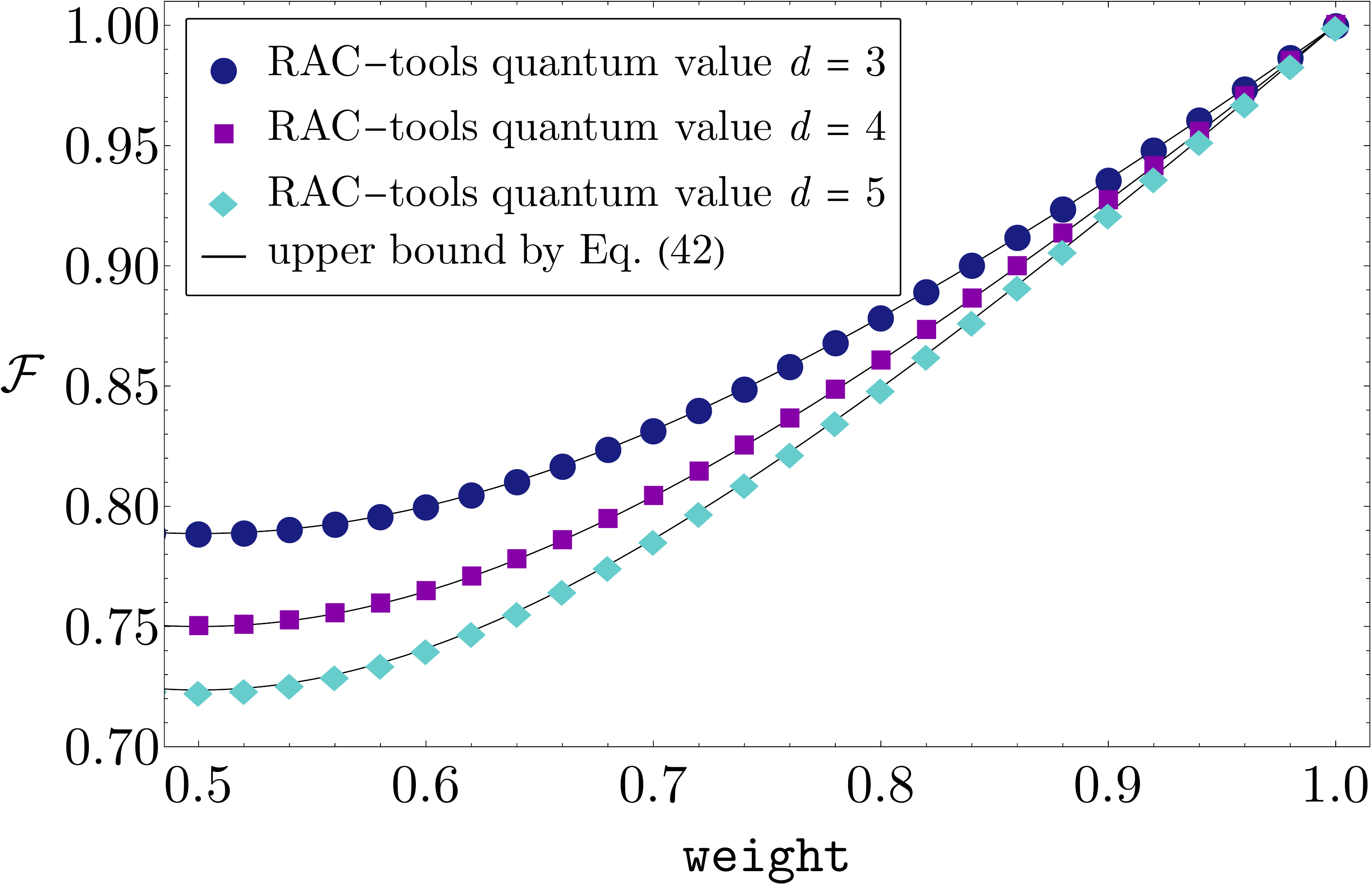}
    \caption{Top: Optimal performance over quantum (blue dots) and classical strategies (red squares) of the $2^5\mapsto 1$ $b$-RAC defined by the bias tensor $\alpha_{{\bf x}y}=\frac{1}{d^2}r_y$ corresponding to the \protect\Verb+Y\_ONE+ family, as computed by the RAC-tools package. The numerical results for the quantum value are compared with the upper bound in Eq.~\eqref{upb-d} (solid line). Bottom: Numerical results for the quantum value of the $2^d\mapsto 1$ $b$-RAC for $d=3$ (blue dots), $d=4$ (purple squares), and $d=5$ (cyan diamonds).} \label{fig:Y_ONE_d}
\end{figure}

The upper bound in Eq.~\eqref{upb-d} can also be attained for more general $b$-RACs. Indeed, since for any pair of rank-1 projective measurements it holds that $\tr M_0^{x_0}M_1^{x_1}=|U_{x_0 x_1}|^2$, with $U$ a $d\times d$ unitary matrix, it follows from the condition in Eq.~\eqref{cauchy-schwarz-d} that the  upper bound in Eq.~\eqref{upb-d} will be attainable with rank-1 projectors if there exist a unistochastic matrix $B$ satisfying
\begin{equation}
    B_{x_0 x_1}=1+\frac{1}{4r_0r_1}(C^2\alpha_{x_0 x_1}^2-1).
\end{equation}
Lastly, it is worth remarking that for some particular biases the optimal measurement operators may satisfy ${\rm rank}(M_{y}^{x_y})\neq 1$, as suggested by Eq.~\eqref{rank}. Actually, as follows from the discussion above, for any pair of projective measurements saturating inequality Eq.~\eqref{lambda_max} we can find $b$-RAC such that its optimal value is attained by these measurements. By summing over $x_0$ and $x_1$ in Eq.~\eqref{cauchy-schwarz-d} we find that the entries of the bias tensor specifying this $b$-RAC are given by
\begin{equation}
    \alpha_{x_0x_1}=\frac{\sqrt{1+4r_0r_1(\tr M_0^{x_0}M_1^{x_1}-1)}}{\sum_{x_0x_1} \sqrt{1+4r_0r_1(\tr M_0^{x_0}M_1^{x_1}-1)} }.
\end{equation}

It should be noted that the $b$-RACs defined in this way are not necessarily interesting from the perspective of studying the advantages of quantum resources, since there is no guarantee regarding the distance of the upper bound to the classical value; e.g., if the measurements we have chosen commute, then the upper bound will coincide with the classical value. This procedure can be used to build $b$-RACs tailored to specific pairs of projective measurements, in which the operators' rank is not restricted to 1.

\section{Conclusions}

In this work, we have presented $b$-RACS as a generalization of the RAC protocol in which the distribution of inputs to the parties is not necessarily uniform. Introducing a bias on these distributions has a profound impact on both the optimal value of the RAC functional and encoding-decoding strategies achieving it, and also on the capacity of quantum devices to provide an advantage in the protocol performance. Understanding how to optimize the performance of a given biased RAC is therefore a step in improving our understanding of the advantages of quantum resources.

The problem of optimizing the performance of an arbitrary $b$-RAC can be approached numerically with the aid of the algorithms we have presented here, which can be implemented by means of the RAC-tools Python package we produced for that purpose. The package allows the user to define arbitrary biases in the input distribution, and compute the classical and quantum value of the ensuing RAC functional, along with the encoding and decoding strategies attaining these values. We have used the package to study the $b$-RAC performance for different biases in the $n^2\mapsto 1$ and $2^d\mapsto 1$ scenarios, focused in the case of uncorrelated inputs. For these examples we also provide analytical results for the quantum value and the measurements attaining it, showing how these are determined by the chosen input bias.

In the $n^2\mapsto 1$ scenario, we have found that both classical and quantum optimal strategies may actually ignore part of the input strings. For quantum strategies it is first observed that optimal decoding can always be done with projective measurements. This allows the derivation of a simple upper bound which coincides with the quantum value for $2^2\mapsto 1$ $b$-RACS and, in some cases, for $3^2\mapsto 1$ $b$-RACs. Moreover, it is shown that attaining this upper bound self-tests the angles between the optimal measurement operators and, in particular, for the case of uniformly distributed input strings the optimal $b$-RAC performance certifies that the measurements correspond to MUB's. The argument in the derivation of this upper bound can be extended to the $2^d\mapsto 1$ scenario, providing thus an upper bound to the optimal performance achievable with projective measurements. This bound is shown to be always attainable using mutually unbiased measurements if the distribution of input strings is unbiased, regardless of the bias on the distribution of requested characters. For more general biases the upper bound will in general not be attainable, but we have shown that there are several instances in which this value is achievable. It is still not clear at the time of writing if, as suggested by our numerical results, the optimal $b$-RAC performance in this scenario is attainable only with projective measurements. In that case the upper bound we derived would coincide with the quantum value, and it would be worth to investigate the possibility of extending the self-testing results previously derived for the unbiased RAC in this scenario.

We have focused the discussion of analytical results in this work, almost completely, on the case of biased RACs in which the inputs of both parties are uncorrelated, since introducing correlations between them departs from the original spirit of the RAC protocol. Nevertheless, investigating how correlations in the inputs affect the performance of the protocol is an interesting next step for which the numerical tools we have developed are applicable.

\section{Acknowledgements}

We acknowledge fruitful discussions with M\'at\'e Farkas. The project ``Fundamental aspects of the quantum set of correlations'' (Grant No. 2019/35/D/ST2/02014) is carried out within the SONATA project of the National Science Centre, Poland.


\bibliographystyle{apsrev4-1}
\bibliography{references.bib}

\newpage

\appendix

\section{Proof of Lemma \ref{EXHAUSTIVE SEARCH LEMMA}} \label{proof of the exhaustive search lemma}

\begin{proof} For statement (a), assume that a particular encoding function $E (\mathbf{x}) = \mu$ is fixed. If so, the performance of $\mathcal{F}$ in Eq.~\eqref{general functional} for an arbitrary set of decoding functions $\{D_y\}_{y=0}^{n-1}$ is given by
\begin{align}
\mathcal{F} = \sum_{\mathbf{x}, y} \alpha_{\mathbf{x}, y, D_y ( E (\mathbf{x}))}
= \sum_{\mathbf{x}, y, \mu} \alpha_{\mathbf{x}, y, D_y (\mu)} \, \delta_{\mu, E (\mathbf{x})},
\end{align}
where we used Eq.~\eqref{deterministic behaviors} to compute the statistics. Then, the maximization of $\mathcal{F}$ over the set of decoding functions, is equivalent to the maximization of the image of $D_y (\mu)$:
\begin{align}
\max_{\{D_y\}_y} \mathcal{F} = \sum_{\mu, y} \max_{D_y (\mu)} \bigg\{ \sum_\mathbf{x} \alpha_{\mathbf{x}, y, D_y (\mu)} \, \delta_{\mu, E ( \mathbf{x} )}
\bigg\},
\label{decoding maximization}
\end{align}
which yields, for some $y$ and $\mu$, $D_y^*(\mu) = b$, where $b$ is the optimal image of $D_y(\mu)$.

Now, for statement (b), we proceed similarly by assuming that the decoding functions $D_y (\mu)$ are fixed for all $y$. Then, the value assumed by $\mathcal{F}$ for an arbitrary encoding function $E(\mathbf{x})$ is written as
\begin{align}
\mathcal{F} = \sum_{\mathbf{x}, y} \alpha_{\mathbf{x}, y, D_y ( E(\mathbf{x}) )}.
\end{align}
Analogous to statement (a), the maximization of $\mathcal{F}$ over the encoding function is equivalent to maximize the image of $E(\mathbf{x})$:
\begin{align}
\max_E \mathcal{F} = \sum_\mathbf{x} \max_{E(\mathbf{x})} \bigg\{ 
\sum_y \alpha_{\mathbf{x}, y, D_y (E(\mathbf{x}))}
\bigg\},
\label{encoding maximization}
\end{align}
which produces $E^*(\mathbf{x})=\mu$, where $\mu$ is the optimal image for $E(\mathbf{x})$.
\end{proof}
As explained in the main text, this lemma is useful in reducing the inherent complexity associated with the exhaustive search algorithm. To provide further clarity, we can maximize Eqs.~\eqref{decoding maximization} and \eqref{encoding maximization} with respect to the encoding and decoding functions, respectively. This additional maximization yields the classical value of $\mathcal{F}$ for both case, as follows:
\begin{subequations}
\begin{align}
\mathcal{F}_C = \max_{ E } \bigg\{ 
    \sum_{\mathbf{x}, y} \alpha_{\mathbf{x}, y, D_y^*( E (\mathbf{x}) )}
                            \bigg\}, \label{alternative method 2} \\
\mathcal{F}_C = \max_{ \{ D_y \}_{y=0}^{n - 1} } \bigg\{ 
    \sum_{\mathbf{x}, y} \alpha_{\mathbf{x}, y, D_y ( E^*(\mathbf{x}) )}
                            \bigg\}. \label{alternative method 1}
\end{align}
\end{subequations}
That is, Eqs.~\eqref{alternative method 2} and \eqref{alternative method 1} introduce a two-step maximization method that yields the precise classical value. For Eq.~\eqref{alternative method 2}, we first optimize over the decoding functions and then over the encoding functions. Conversely, for Eq.~\eqref{alternative method 1}, we follow the reverse order. This simple modification avoids the maximization over all combinations of $E$ and $\{ D_y \}_{y=0}^{n - 1}$. Furthermore, since the RAC protocol is asymmetrical with respect to Alice and Bob, the difference between Eqs.~\eqref{alternative method 2} and \eqref{alternative method 1} relies only on the computational complexity for each case.


\section{RAC-tools user guide} \label{Ap:pack}

In the main text, we introduced the functions that make up the RAC-tools package. In this appendix, we provide a more detailed description of these functions and their features.

\subsection{The \texorpdfstring{\protect\Verb+generate\_bias+}{} function} \label{Ap:gen_fun}

Since our interest in this work is to study the quantum and classical value of biased RACs, the main feature of the RAC-tools package is that it allows the user to introduce bias in the RAC functional, which will be optimized by either the \verb+perform_search+ or \verb+perform_seesaw+ functions. One way of doing this is by building a custom bias tensor and passing it as an argument to either of these functions as a Python dictionary. However, as constructing a bias tensor requires some effort, we provide a functionality that allows the user to choose from several simple and natural families of bias tensors. Each of these families takes one or more parameters, and it is particularly interesting to study the behavior of RACs as we vary the parameter. An example of such a construction was given at the beginning of Section~\ref{S:package}. In what follows, we describe in detail the built-in bias options that the user can access via the \verb+generate_bias+ function.

The goal of \verb+generate_bias+, in short, is to construct a properly normalized bias tensor using only a few previously specified parameters. 
This function is not intended to be called by the user, who should in turn specify the parameters defining the desired bias tensor as arguments of the optimization functions. In order to do so, the value of two variables, \verb+bias+ and \verb+weight+, must be specified. The variable \verb+bias+ is a string determining the structure of the bias to be generated, whereas the variable \verb+weight+ is a real-valued parameter (or a vector of parameters) that determines the actual weights given to different terms in the objective function.

As an example, we can consider a general version of the \verb+Y_ONE+ bias family already introduced in the main text. This is a family of bias tensors in which the input strings $\bf x$ are distributed uniformly, but there is bias in Bob's input, as one of the characters of $\bf x$, e.g., $x_k$, is requested more (or less) frequently than the others. If we call $w$ the parameter defining how often Bob is asked to recover $x_k$, then the bias tensor takes the form
\begin{align}
\alpha_{\mathbf{x} y} = 
\begin{cases}
    \frac{1}{m^n} w \, & \text{if} ~ y = 0, \\
    \frac{1}{m^n}\frac{(1 - w)}{n - 1} & \text{otherwise}.
    \label{YONE bias}
\end{cases}
\end{align}
In order to build this bias tensor via the \verb+generate_bias+ function, we need to pass as arguments of either \verb+perform_search+ or \verb+perform_seesaw+ the following string and float: \verb+bias="Y_ONE"+ and  \verb+weight=w+. By symmetry, the \verb+Y_ONE+ family considers only biasing the first character against the rest, as biasing other values of $y$ produces analogous results. It is possible, nevertheless, to introduce a bias on the frequency with which any of the characters $x_y$ is requested from Bob. This can be done by setting \verb+bias="Y_ALL"+ and \verb+weight=List+, where \verb+List+ is a list (or a tuple) of floats of length $n$ adding up to one. In this case, the bias tensor obtained from \verb+generate_bias+ takes the form
\begin{align}
\alpha_{\mathbf{x} y} =  \frac{1}{m^n} w_y,
\end{align}
where $w_y$ is the weight corresponding to the $y$-th character in the input string $\bf x$ and the factor $\frac{1}{m^n}$ results from the input strings $\bf x$ being uniformly distributed.

For introducing biases in the distribution of input strings, the package offers several one-parameter families, which we enumerate below:

\begin{enumerate}
    \item \verb+X_ONE+. Analogous to the \verb+Y_ONE+ family, it biases the input $\mathbf{x} = 0^{\times n}$ against the $m^n-1$ remaining strings. The user is allowed to define the weight $w$ that will be given to this first input, which will be used to generate a bias tensor of the form $\alpha_{{\bf x} y}=\alpha_{\bf x}\frac{1}{n}$, where
    \begin{equation}
        \alpha_{\bf x}=\begin{cases}
                        w & \text{if}\; {\bf x}=0^{\times n}, \\
                        \frac{1 - w}{m^n - 1} &\text{otherwise.}
                       \end{cases}
    \end{equation}
    \item \verb+X_DIAG+. This family of biases gives a special weight to input strings of the form $\mathbf{x} = i^{\times n}$, where $i = 0,\, ...,\, m - 1$. Since there are $m$ of these strings, in terms of the parameter $w$ controlled by the user, the distribution of input strings takes the form
    \begin{equation}
        \alpha_{\bf x}=\begin{cases}
                        \frac{w}{m} & \text{if}\; {\bf x}=i^{\times n}, \\
                        \frac{1 - w}{m^n - m} & \text{otherwise.}
                       \end{cases}
    \end{equation}
    \item \verb+X_CHESS+. In this case, the input strings are split into two classes depending on whether $\sum_j x_{j}$ is odd or even. Since the parity of the total number of strings is the same as that of $m$, when the latter is even, half of the strings go into each of the classes defined above. In that case, in terms of the weight $w$ chosen by the user, the ensuing distribution of input strings is given by
    \begin{equation}
    \alpha_{\bf x}=\begin{cases}
                    \frac{2w}{m^n} & \text{if}\; \sum_j x_{j}\;\text{is even}, \\
                    \frac{2(1 - w)}{m^n} & \text{otherwise}.
                   \end{cases}
    \end{equation}
    On the other hand, if $m$ is odd, the number of strings satisfying $\sum_j x_{j}$ odd is $\frac{m^n - 1}{2}$. In that case, the distribution of input strings reads 
    \begin{equation}
    \alpha_{\bf x}=\begin{cases}
                    \frac{2w}{m^n - 1} & \text{if}\; \sum_j x_{j}\;\text{is odd}, \\
                    \frac{2(1 - w)}{m^n + 1} & \text{otherwise}.
                   \end{cases}
    \end{equation}
    For $n = 2$, we can think of the elements of $\alpha_{\bf x}$ as the entries of a matrix, in which case the biased elements are arranged in a pattern that resembles a chess board.
    
    \item \verb+X_PLANE+. As before, the idea of this type of bias is to split the set of strings into two classes defined by the condition $x_0=0$. This corresponds to biasing just the first bit of the input string. Since there are $m^{n - 1}$ strings satisfying $x_0=0$, in terms of the parameter $w$, the ensuing distribution of input strings reads
    \begin{equation}
    \alpha_{\bf x}=\begin{cases}
                    \frac{w}{m^{n - 1}} & \text{if}\; x_{0}=0, \\
                    \frac{1 - w}{m^n - m^{n - 1}} & \text{otherwise}.
                   \end{cases}
    \end{equation}
\end{enumerate}

All the biases introduced so far depend only on $\mathbf{x}$ or only on $y$. In the next step, we could take one bias of each kind and combine them, which would lead to a product distribution over $\mathbf{x}$ and $y$. However, as mentioned in the introduction, there is no reason why we should restrict ourselves to product distributions. If we go back to the linear functional given in Eq.~\eqref{general functional}, it is natural to consider the case where the coefficients of the functional depend only on $b$, i.e., the answer that Bob is expected to give. In Appendix~\ref{Ap:sec4}, we discuss such scenarios and refer to them as \verb+B_ONE+ and \verb+B_ALL+ biases. In our usual language such cases correspond to nonfactorizable distributions of inputs $\mathbf{x}$ and $y$. The first bias of this kind, \verb+bias="B_ONE"+, corresponds to biasing the first outcome of Bob, $b=0$, against the remaining $d-1$ outputs
%
%
\begin{align}
\alpha_{\mathbf{x} y} = 
\begin{cases}
    \frac{1}{n} \frac{1}{m^{n - 1}} \, w \, & \text{if} ~ x_y = 0, \\
    \frac{1}{n} \frac{1}{m^{n - 1}}\frac{(1 - w)}{m - 1} & \text{otherwise},
\end{cases}
\end{align}
where $\frac{1}{n}\frac{1}{m^{n-1}}$ is the normalization factor. If a more general bias of the outputs is required, then the user can enter \verb+bias="B_ALL"+ as an option, in which case they should input as weight a Python list (or tuple). The \verb+generate_bias+ function will then output a bias tensor of the form
\begin{align}
\alpha_{\mathbf{x} y} = \frac{1}{n}\frac{1}{m^{n - 1}} \, w_{x_y},
\end{align}
where $w_{x_y}$ is the weight on the character $x_y$ -- and consequently on the $b$-th output of Bob since $b=x_y$.


\subsection{The \texorpdfstring{\protect\Verb+perform\_search+}{} function} \label{Ap:class_val_fun}

The goal of \verb+perform_search+ is to exactly compute the best classical performance of a given $n^m$ {\fontsize{9.5}{104}$\smash{\overset{d}{\mapsto}} $} 1 RAC. The function can perform this computation either via a complete exhaustive search or by means of the less expensive approach that follows from Eqs.~\eqref{alternative method 2} and \eqref{alternative method 1}. To operate \verb+perform_search+ it is enough to specify in its argument the three integers defining the scenario, $n$, $d$, and $m$, and the search method. The latter can be introduced by declaring either \verb+method=0+ for a pure exhaustive search, \verb+method=1+ for a search over encoding maps as in Eq.~\eqref{alternative method 2}, or \verb+method=2+ for a search over decoding maps as in Eq.~\eqref{alternative method 1}, which is the default method. Furthermore, the value of $m$ is set by default to coincide with that of $d$, so that users are not expected to declare it unless they require these numbers to be different.

An example of how this function operates can be seen in FIG.~\ref{classical example appendix}, in which the user desires to estimate the classical value of the $2^2 \mapsto 1$ unbiased RAC. The function is called passing as arguments \verb+n=2+, \verb+d=2+, and \verb+method=0+, and once the procedure is finished the report in FIG.~\ref{classical example appendix} is printed. The \emph{Summary of computation} section of the report informs the user the total time of computation as well as the total number of encoding and decoding functions analyzed for the chosen search method. For the case of \verb+method=0+, this latter information corresponds to the total number of combinations of encoding and decoding functions, i.e., $d^{m^n} \times m^{dn}$. In addition, the average time taken to iterate over each function (or combination of encoding and decoding functions, if \verb+method=0+) is displayed at \emph{Average time per function}.

\begin{figure}[hb]
\centering
\footnotesize
\begin{BVerbatim}[commandchars=\\\{\}]
> perform_search(n=2, d=2, method=0)

==========================================================
                      RAC-tools v1.0
==========================================================

----------------- Summary of computation -----------------

Total time of computation: 0.001859 s
Total number of encoding/decoding functions: 256
Average time per function: 7e-06 s

----------- Analysis of the optimal realization ----------

Computation of the classical value for the 2\textsuperscript{2}-->1 RAC:
0.75
Number of functions achieving the computed value: 24

First functions found achieving the computed value
Encoding: 
E: [0, 0, 0, 1]
Decoding: 
D\textsubscript{0}: [0, 1]
D\textsubscript{1}: [0, 1]

------------------- End of computation -------------------
\end{BVerbatim}
\caption{Report produced by the \protect\Verb+perform\_search+ function for the unbiased $2^2 \mapsto 1$ RAC. The first part of the report provides the user with information about the computation, whereas the second part provides the user with the RAC classical value and a realization attaining it. In this example, the value is found through an exhaustive search of the best value over both encoding and decoding maps.
\label{classical example appendix}}
\end{figure}

In the second part of the report, the user can see the computed classical value and the number of functions that achieve this value. Also, the report provides the user with a particular pair of encoding and decoding strategies attaining the optimal value. For the encoding function $E(\mathbf{x})$, the result is displayed in a tuple that is organized in ascending order of $\mathbf{x}$, i.e., $[E(00...0),~ E(00...1),~ ...,~ E((m-1)...(m-1))]$. For the decoding functions $D_y(\mu)$, each row corresponds to a distinct input $y$ and it is organized in ascending order of $\mu$. As it is expected from Lemma~\ref{lem_optclass}, the optimal strategy reported in FIG.~\ref{classical example appendix} consists of a majority encoding function and identity map for decoding.

Before closing, we would like to recall that the exhaustive search method will necessarily require more computation time than that required by either of the two approaches following from Lemma~\ref{EXHAUSTIVE SEARCH LEMMA}, since in the first case the search is performed over both encoding and decoding maps. TABLE~\ref{table:working cases} presents a comparison in terms of computation time for all of the three methods. The table shows that the method in Eq.~\eqref{alternative method 1} is the best in terms of computation time for most part of the cases, except for a few cases in which $d$ is the largest integer and $n=2$. In those cases, the method in Eq.~\eqref{alternative method 2} is equivalent or better.

\setlength{\tabcolsep}{6 pt}
\begin{table*}[t]
    \centering
    \caption{Cases which can be executed in less than one hour for \protect\Verb+method=0+. Computation time comparison with \protect\Verb+method=1+ and \protect\Verb+method=2+. For the construction of this table, an octa-core CPU was used (4 $\times$ 3.2 GHz and 4 $\times$ 2.064 GHz). The first column contains the executable cases for a given $n^m \smash{\overset{d}{\mapsto}} 1$ RAC. The second column contains the classical value computed for each unbiased case. The other columns contain the time taken to execute each procedure. Cases are sorted in ascending order by computation time for \protect\Verb+method=0+.}
    \begin{tabular}{||c c c c c||}
        \hline
        $n^m\overset{d}{\mapsto} 1$ & Classical value & Time (\protect\Verb+method=0+) & Time (\protect\Verb+method=1+) & Time (\protect\Verb+method=2+) \\ [0.5 ex] 
        \hline\hline
        $2^2\overset{2}{\mapsto} 1$ & $\sfrac{3}{4}$ & 1.8  ms & 0.32 ms & 0.28 ms\\ [0.15 ex]
        \hline
        $2^2\overset{3}{\mapsto} 1$ & $\sfrac{7}{8}$ & 26 ms & 2 ms & 1.2 ms \\ [0.15 ex]
        \hline
        $3^2\overset{2}{\mapsto} 1$ & $\sfrac{3}{4}$ & 0.14 s & 11 ms & 2.4 ms \\ [0.15 ex]
        \hline
        $2^2\overset{4}{\mapsto} 1$ & 1 & 0.18 s & 7.8 ms & 6.5 ms \\ [0.15 ex]
        \hline
        $2^3\overset{2}{\mapsto} 1$ & $\sfrac{5}{9}$ & 0.22 s & 21 ms & 2.6 ms \\ [0.15 ex]
        \hline
        $2^2\overset{5}{\mapsto} 1$ & 1 & 1.4 s & 19 ms & 28 ms \\ [0.15 ex]
        \hline
        $2^2\overset{6}{\mapsto} 1$ & 1 & 12 s & 40 ms & 82 ms \\ [0.15 ex]
        \hline
        $3^2\overset{3}{\mapsto} 1$ & $\sfrac{19}{24}$ & 21 s & 0.16 s & 23 ms\\ [0.15 ex]
        \hline
        $2^3\overset{3}{\mapsto} 1$ & $\sfrac{2}{3}$ & 1 min 4 s & 0.42 s & 27 ms \\ [0.15 ex]
        \hline
        $2^2\overset{7}{\mapsto} 1$ & 1 & 1 min 29 s & 61 ms & 0.25 s \\ [0.15 ex]
        \hline
        $2^4\overset{2}{\mapsto} 1$ & $\sfrac{7}{16}$ & 2 min 20 s & 2.2 s & 12 ms \\ [0.15 ex]
        \hline
        $4^2\overset{2}{\mapsto} 1$ & $\sfrac{11}{16}$ & 5 min 18 s & 2.4 s & 21 ms \\ [0.15 ex] 
        \hline
        $2^2\overset{8}{\mapsto} 1$ & 1 & 11 min 6 s & 0.1 s & 1 s \\ [0.15 ex]
        \hline
        $3^2\overset{4}{\mapsto} 1$ & $\sfrac{5}{6}$ & 33 min 16 s & 1.5 s & 0.12 s \\ [0.75 ex]
        \hline
    \end{tabular}
    \label{table:working cases}
\end{table*}

\subsection{The \texorpdfstring{\protect\Verb+perform\_seesaw+}{} function} \label{Ap:opt_val_fun}

This function implements the see-saw algorithm described in Section \ref{see-saw method}, and its goal is to provide lower bounds to the quantum value of a given $n^m$ {\fontsize{9.5}{104}$\smash{\overset{d}{ \mapsto} } $} $1$ $b$-RAC. As is the case with \verb+perform_search+, the \verb+perform_seesaw+ function takes as argument the integers defining the scenario, $n$, $d$, and $m$ and the bias tensor, either as a dictionary or via one of the aforementioned built-in options. The user is also asked to pass as an argument the number of starting points for the algorithm by means of the variable \verb+seeds+. Moreover, it is possible to use this function to compute a lower bound to the classical value, by means of the variable \verb+diagonal+. If \verb+diagonal=True+, the function initializes the see-saw algorithm with random diagonal measurements, and the optimization is then restricted to operators which are diagonal in the computational basis. By default, \verb+diagonal=False+, and the algorithm optimizes the functional value over POVM measurements.


When called, \verb+perform_seesaw+ runs the see-saw algorithm as many times as the number of seeds specified by the user, generating a lower bound to the quantum value per starting point. The best value is therefore the largest among all these lower bounds, implying that the chances of the function providing the actual quantum value of the $b$-RAC increase with the number of seeds, as well as the computation time. In TABLE~\ref{table:number of seeds}, we provide the number of seeds used to generate the numerical results presented in the main text and later in Appendix \ref{Ap:sec4}.

\begin{table}[b]
    \centering
    \caption{Number of seeds used for the numerical results presented in the figures. The first column contains the executable cases for a given $n^d \mapsto 1$ RAC followed by the bias family in the second column. The third column contains the number of seeds used followed by a link to the respective figure. For the cases where the realization varies smoothly according to the \protect\Verb+weight+, only a few seeds are needed. This is the case of all bias families explored in this work expect for \protect\Verb+X\_PLANE+, in which there is a critical point for \protect\Verb+weight+ in which the realization starts to ignore a bit.}
    \begin{tabular}{||c c c c||}
        \hline
        $n^d \mapsto 1$ & \protect\Verb+bias+ & Seeds & Figure\\ [0.5 ex] 
        \hline\hline
        $3^2 \mapsto 1$ & \protect\Verb+X_ONE+ & 3 & FIG.~\ref{fig:X_ONE}\\ [0.1 ex]
        \hline
        $3^2 \mapsto 1$ & \protect\Verb+X_ONE+ with $r_0=0.5$ & 3 & FIG.~\ref{fig:X_ONE_Y_ONE}\\
         & and $r_1 = r_2 = 0.25$ &  & \\ [0.1 ex]
        \hline
        $2^3 \mapsto 1$ & \protect\Verb+Y_ONE+ & 3 & FIG.~\ref{fig:Y_ONE_d}\\ [0.1 ex]
        \hline
        $2^4 \mapsto 1$ & \protect\Verb+Y_ONE+ & 3 & FIG.~\ref{fig:Y_ONE_d}\\ [0.1 ex]
        \hline
        $2^5 \mapsto 1$ & \protect\Verb+Y_ONE+ & 3 & FIG.~\ref{fig:Y_ONE_d}\\ [0.1 ex]
        \hline
        $4^2 \mapsto 1$ & \protect\Verb+Y_ONE+ & 3 & FIG.~\ref{fig:Y_ONE(n=4)}\\ [0.1 ex]
        \hline
        $2^2 \mapsto 1$ & \protect\Verb+X_CHESS+ & 3 & FIG.~\ref{fig:X_CHESS_2}\\ [0.1 ex]
        \hline
        $2^2 \mapsto 1$ & \protect\Verb+X_PLANE+ & 10 & FIG.~\ref{fig:x_plane2}\\ [0.1 ex]
        \hline
        $3^2 \mapsto 1$ & \protect\Verb+X_PLANE+ & 10 & FIG.~\ref{fig:x_plane3}\\ [0.1 ex]
        \hline
    \end{tabular}
    \label{table:number of seeds}
\end{table}

 
Because the see-saw algorithm is iterative, convergence criteria must be adopted to decide whether the optimal value for a given seed has been attained after a particular number of steps. In the \verb+perform_seesaw+ implementation of this algorithm, we impose two convergence criteria, and the procedure is finished whenever the two are satisfied. The first criterion is related to the convergence of the $\mathcal{F}$ value. It is satisfied whenever the difference between two consecutive evaluations of $\mathcal{F}$ is smaller than a value that can be set by the user via the variable \verb+prob_bound+. The default value of this variable is set to $10^{-9}$. The second stopping criterion considers the convergence of the measurements, and it focuses on the distance between the optimal measurement operators in two consecutive iterations of the algorithm. More precisely, we will say that the measurements converged if the condition
\begin{align}
\max_{y, b} \big|\big| M_y^b - N_y^b \big|\big| < t \label{meas_bnd} 
\end{align}
is satisfied, where $|| \cdot ||$ denotes the Frobenius norm, $N_y^b$ and $\smash{M_y^b}$ denote two consecutive measurement operators associated with the same value of the $y$-th character of the input $\bf x$, and $t$ is a threshold that can be defined by the user via the variable \verb+meas_bound+, which as a default takes the value $10^{-7}$. For the evaluation of the condition in Eq.~\eqref{meas_bnd}, we use the function \verb+norm+, from \verb+numpy.linalg+, to implement the Frobenius norm.
    
\begin{figure}[b]
\centering
\footnotesize
\begin{BVerbatim}
> perform_seesaw(n=2, d=2, seeds=5)

==========================================================
                      RAC-tools v1.0
==========================================================

----------------- Summary of computation -----------------

Number of random seeds: 5
Average time for each seed: 0.14852 s
Average number of iterations: 3
Seeds 1e-13 close to the best value: 5

----- Analysis of the optimal realization for seed #1 ----

Estimation of the quantum value for the 2²-->1 QRAC: 
0.853553390593

Measurement operator ranks
M[0] ranks:  1  1
M[1] ranks:  1  1

Measurement operator projectiveness
M[0, 0]:  Projective		6.44e-15
M[0, 1]:  Projective		6.44e-15
M[1, 0]:  Projective		6.78e-15
M[1, 1]:  Projective		6.78e-15
 
Mutual unbiasedness of measurements
M[0] and M[1]:  MUB		5.91e-14

------------------- End of computation -------------------
\end{BVerbatim}
\caption{Report produced by the function \protect\Verb+perform\_seesaw+. In the first part of the report, the function produces a summary of the computation, displaying information such as the number of random starting points, etc. In the second part, it displays the largest optimal value found among all seeds and an analysis of the optimal measurements obtained by this seed.}
\label{simple example appendix}
\end{figure}

The value of both, \verb+prob_bound+ and \verb+meas_bound+, can be passed as an argument to \verb+perform_seesaw+. In addition to the convergence criteria, we have imposed a limit to the number of iterations to be executed by the algorithm, so that if after 200 iterations either the value or the measurements fail to converge, the calculation stops. In this case, the message \emph{maximum number of iterations reached} is displayed as a warning. This limit can be modified by entering a different value to the variable \verb+max_iterations+ in the argument of the function.


An example of the operation of \verb+perform_seesaw+ can be seen in FIG.~\ref{simple example appendix}, in which the user wants to estimate the quantum value of the $2^2 \mapsto 1$ unbiased RAC. As in the case of \verb+perform_search+, the user passes as arguments $n=2$ and $d=2$ to define the scenario, but now instead of choosing a search method the user introduces the number of starting points to be used by passing \verb+seeds=5+. After finishing the procedure the function prints a report divided into two parts. In the \emph{Summary of the computation}, it presents the number of random starting points, the average processing time, and the average number of iterations among all starting points. In addition, it shows how many starting points produced an optimal value that is close to the largest value obtained. The interval to consider two values produced by different starting points as close is the accuracy of the solver MOSEK, which is set to $10^{-13}$. This informs the user how frequent it is to obtain such an estimation; if this number is much smaller than \verb+seeds+, this indicates that the user should increase the number of starting points in case of a new execution.  

In the second part, the estimation of the optimal value is reported, followed by information about the set of measurements attaining such value. Note that the reported value in FIG.~\ref{simple example appendix} matches the one found by Ref.~\cite{Ambainis09}. Next, the report displays the rank of the optimal measurement operators, which is computed using the function \verb+matrix_rank+ of \verb+numpy.linalg+. In addition, the user can check whether the measurement operators are projective. The number shown in the second column of \emph{Measurement operator projectiveness} corresponds to the quantity
\vspace{3mm}
\begin{align}
\big|\big|(M_y^b)^2 - M_y^b \big|\big|.
\end{align}
For both of these checks, rank and projectiveness, we preset a tolerance of $10^{-7}$.

Lastly, in the case where at least two measurements are rank-one and projective, the function also computes whether each pair of measurements can be constructed out of mutually unbiased bases (MUB). For a pair of rank-one projective measurements, let us say $\{P^a\}_{a=0}^{m-1}$ and $\{Q^b\}_{b=0}^{m-1}$, where $a$ and $b$ denote the $a$-th and the $b$-th outcome, it is enough \cite[App. B]{Tavakoli21} to check if
\vspace{3mm}
\begin{multline}
P^a = m \, P^a Q^b P^a \;\; \text{and} \;\; Q^b = m \, Q^b P^a Q^b \\
\forall ~ a,\, b \in \{0,\, 1\, ...,\, m-1\}.
\end{multline}
In this case, the number displayed in the second column of \emph{Mutual unbiasedness of measurements} represents the quantity
\vspace{3mm}
\begin{align}
\max_{a, b} \big\{|| m\, P^a Q^b P^a - P^a ||, || m \, Q^b P^a Q^b - Q^b ||\big\}. \label{MUM measure}
\end{align}
For the cases in which the amount in Eq.~\eqref{MUM measure} is lower than \verb+MUB_BOUND=5e-6+, the function prints \verb+MUB+. Otherwise, it simply displays \verb+Not MUB+.

\vspace{2.5mm}

\section{Analysis for other built-in families of bias} \label{Ap:sec4}

In the main text, we have used the analytical results derived for the $2^n\mapsto 1$ scenario to study the quantum value of the $b$-RACs determined by the \protect\verb+X_ONE+ bias family introduced above. Here, we offer a similar analysis for the $b$-RAC families determined by others of these built-in biases in the $2^n\mapsto 1$ scenario.

\vspace{2.5mm}

\subsection{The \texorpdfstring{\protect\Verb+Y\_ONE+}{} bias family}

We start by looking at the case where the bias is only on the requested bit $y$, i.e., $\alpha_{\bf x}=\frac{1}{2^n}$, which leads to $p_{\bf x}=\frac{1}{2^{n-1}}\,\forall {\bf x}$. From $p_{\bf x}$ being constant follows, for $n=2$ and $n=3$, that $\cos{(\theta_{ij})}=0\;\; \forall\, i\neq j$ in Eq.~\eqref{gcond}, i.e., the optimal measurements are mutually unbiased. The quantum value is therefore given by the upper bound in Eq.\eqref{upb},
\begin{equation}
    F_Q=\frac{1}{2} + \frac{1}{2} \sqrt{\sum_{y} r_y^2}.
\end{equation}

For $n=4$, the upper bound is not attainable when $r_y=\frac{1}{4}$, as it would require the four vectors $\{{\bf m}_y\}$ to be mutually orthogonal. For weak biases ($r_y\approx \frac{1}{4}$) satisfying the conditions in Eq.~\eqref{gcond}, it would still require these vectors to be linearly independent, and therefore the upper bound is still not attainable. However, if we consider a stronger bias such that the weight on one of the bits becomes negligible, we would expect the bound in Eq.~\eqref{upb} to be attainable again. We can realize such situation by defining
\begin{equation}
    r_y=\begin{cases}
        w & \text{if} \;  y=0 \\
        \frac{1-w}{3} & \text{otherwise}
        \end{cases}, \quad 0\leq w\leq 1. \label{y_one_bias}
\end{equation}
Clearly in this case $w=0 \Rightarrow r_0=0$, and the value of the first bit is never requested from Bob to be decoded. Thus, there are only three Bloch vectors representing measurements that can be chosen to be orthogonal to each other, so that the bound in Eq.~\eqref{upb} is attained.

\begin{figure}[b!]
    \centering
    \hfill\includegraphics[width = 8.65 cm]{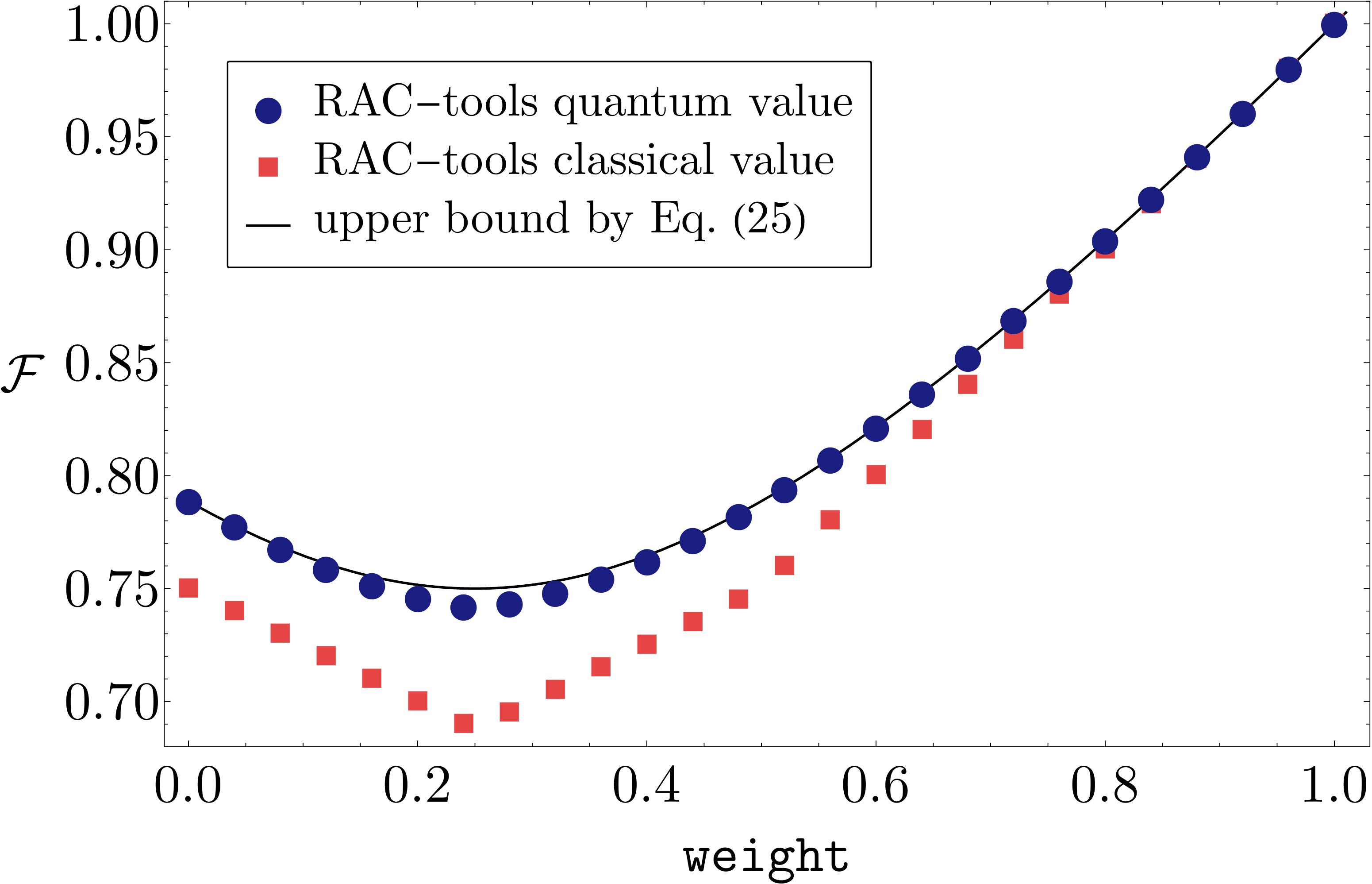}
    \caption{Optimal performance of the $4^2\mapsto 1$ $b$-RAC, with \protect\Verb+Y\_ONE+ bias, over classical (red squares) and quantum (blue dots) encoding-decoding strategies.}\label{fig:Y_ONE(n=4)}
\end{figure}

This is indeed the case, as is shown in FIG.~\ref{fig:Y_ONE(n=4)}, which depicts the results provided by the numerical package for the optimal performance of this $4^2\mapsto 1$ $b$-RAC over both, quantum and classical encoding-decoding strategies. As can be seen in this plot, the numerical quantum values (blue dots) lie very close to  upper bound (solid line), coinciding with it only at the extremal values ($w=0,~1$) and exhibiting the largest difference at $w=\frac{1}{4}$. An inspection of the optimal measurements extracted from the numerical solutions shows that the angles $\theta_{ij}$ parametrizing the measurements naturally divide into two sets, $\Theta_0=\{\theta_{0 i}\}$ and $\Theta_{\bar 0}=\{\theta_{ij}\}$, with $i\neq j=1, 2, 3$, and the angles in one of them being all $\frac{\pi}{2}$. As explained above, for $r_0=0$, the optimal solution involves only the vectors ${\bf m}_i$, $i=1,~2,~3$, which can be chosen to be mutually orthogonal. For small values of $r_0$, we would then expect these three vectors to remain orthogonal (or close to orthogonal), since they contribute the most to the functional value, and ${\bf m}_0$ to be some linear combination of them. The solution provided by the numerical search shows that this intuition is correct since for $w\in[0, w_0]$, with $w_0\approx 0.27415$, we find $\{{\bf m}_i\}$ to be an orthogonal set, with ${\bf m}_0$ being aligned with any of them, meaning that the angles in $\Theta_{\bar 0}$ are all $\frac{\pi}{2}$. This solution has also been found by the authors of Ref.\cite{Ambainis09} in a numerical search for the case of $w=\frac{1}{4}$.

On the other hand, if $w=r_0=1$, only the first bit is to be retrieved by Bob, which can be done with probability $1$ with a classical strategy. For $w\approx 1$ therefore we would expect ${\bf m}_0$ to lie orthogonal to the subspace spanned by $\{{\bf m}_i\}$, which is therefore bound to have dimension $2$. This is indeed the case, as shown by our numerical results: For $w\in[w_0, 1]$ the optimal value is numerically attained with a decoding strategy in which ${\bf m}_0$ is orthogonal to all the ${\bf m}_i$'s, which therefore span a plane in $\mathbb{R}^3$, implying that the angles in $\Theta_{0}$ are all $\frac{\pi}{2}$, while the angles in $\Theta_{\bar 0}$ can be chosen to be $\Theta_{\bar 0}=\{\frac{\pi}{3}\}$. It follows then that the ${\bf m}_i$ are uniformly distributed in the plane orthogonal to ${\bf m}_0$.

\subsection{The \texorpdfstring{\protect\Verb+X\_CHESS+}{} bias family}

Let us now consider a different distribution $\alpha_{\bf x}$ for the input strings, given by
\begin{equation}
    \alpha_{\bf x}=\begin{cases}
        \frac{w}{2^n} & \text{if} \;  \sum_{i} x_i\; \text{odd} \\
        \frac{1-w}{2^n} & \text{otherwise}
        \end{cases}, \quad 0\leq w\leq 1 \label{x_chess_bias}
\end{equation}
and an arbitrary distribution $\{r_y\}$ for the requested bit. As we will show now, this bias has no net effect on the $b$-RAC value when the number of bits is odd. Indeed, note first that there are $2^n$ input bit strings, half of which are such that $\sum_{i} x_i$ is even. Now if $\tilde{\bf x}$ is the string obtained from $\bf x$ by flipping all of its bits, then the sum of bits has the same parity in both strings if $n$ is even, whereas if $n$ is odd this parity is different. As a result, it follows from Eq.~\eqref{x_chess_bias} that for odd $n$ we have $p_{\bf x}=\alpha_{\bf x} + \alpha_{\tilde{\bf x}} = \frac{1}{2^{n-1}}$ and the functional value becomes
\begin{equation}
    F_Q = \frac{1}{2}+\frac{1}{2^{n-1}}\max_{\{{\bf m}_y\}}\sum_{\bf x} \,|\sum_y r_y (-1)^{x_y} {\bf m}_y|, 
\end{equation}
which is the same as that of the unbiased case. We can illustrate this feature by analyzing the $n=2$ and $n=3$ cases. A direct calculation shows that for $n=2$ the value is given by
\begin{widetext}
\begin{align}
F_Q &= \frac{1}{2}+ \max_{\{{\bf m}_0, {\bf m}_1\}}\frac{1-w}{2}|r_0{\bf m}_0+r_1{\bf m}_1| + \frac{w}{2}|r_0{\bf m}_0-r_1{\bf m}_1| \nonumber \\
&\leq \frac{1}{2}+ \max \left\{\frac{1-w}{4}+\frac{w}{4}(r_0-r_1), \frac{w}{4}+\frac{1-w}{4}(r_0-r_1), \frac{1}{2\sqrt{2}}\sqrt{w^2+(1-w^2)}\sqrt{r_0^2+r_1^2} \right \} \label{Q_CHESS_2}
\end{align}
\end{widetext}
where, in the second line, we used Lemma \ref{lemma_2_bit}. In FIG.~\ref{fig:X_CHESS_2}, we show this value as a function of the biasing parameter $w$, compared with the numerical results for the optimal value over both, quantum and classical strategies, for the case of $r_0=r_1=\frac{1}{2}$. As is easy to check from Eq.~\eqref{Q_CHESS_2}, the optimal performance for quantum strategies is better than that over the classical ones for $w\in (0, 1)$, becoming equal only for $w=0\, (1)$ in which case only the strings $00$ and $11$ ($01$ and $10$) are given to Alice for encoding.

\begin{figure}[hb!]
    \centering
    \hfill\includegraphics[width = 8.6 cm]{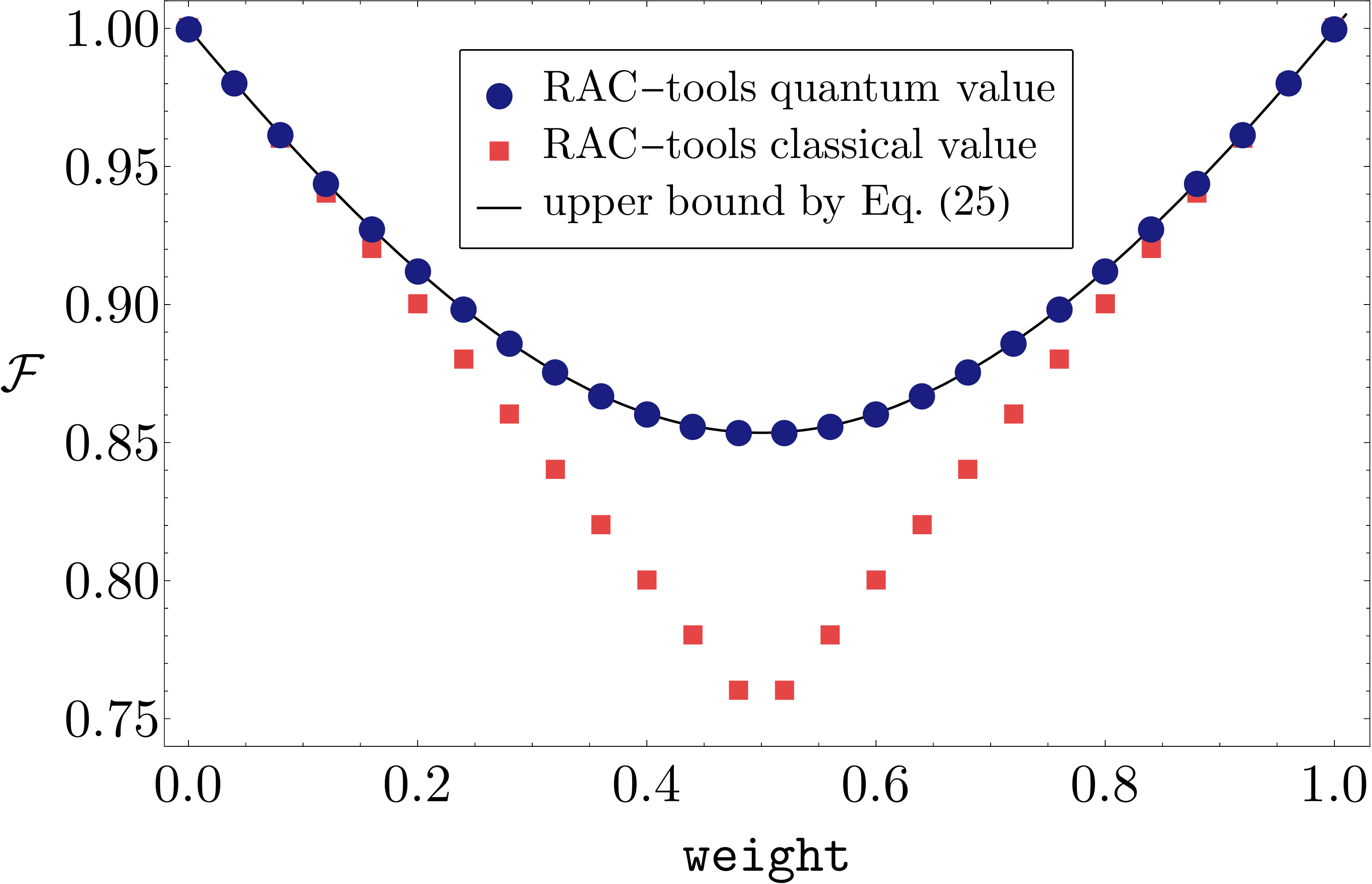}\\
    \vspace{2.5 mm}
    \hfill\includegraphics[width = 8.4 cm]{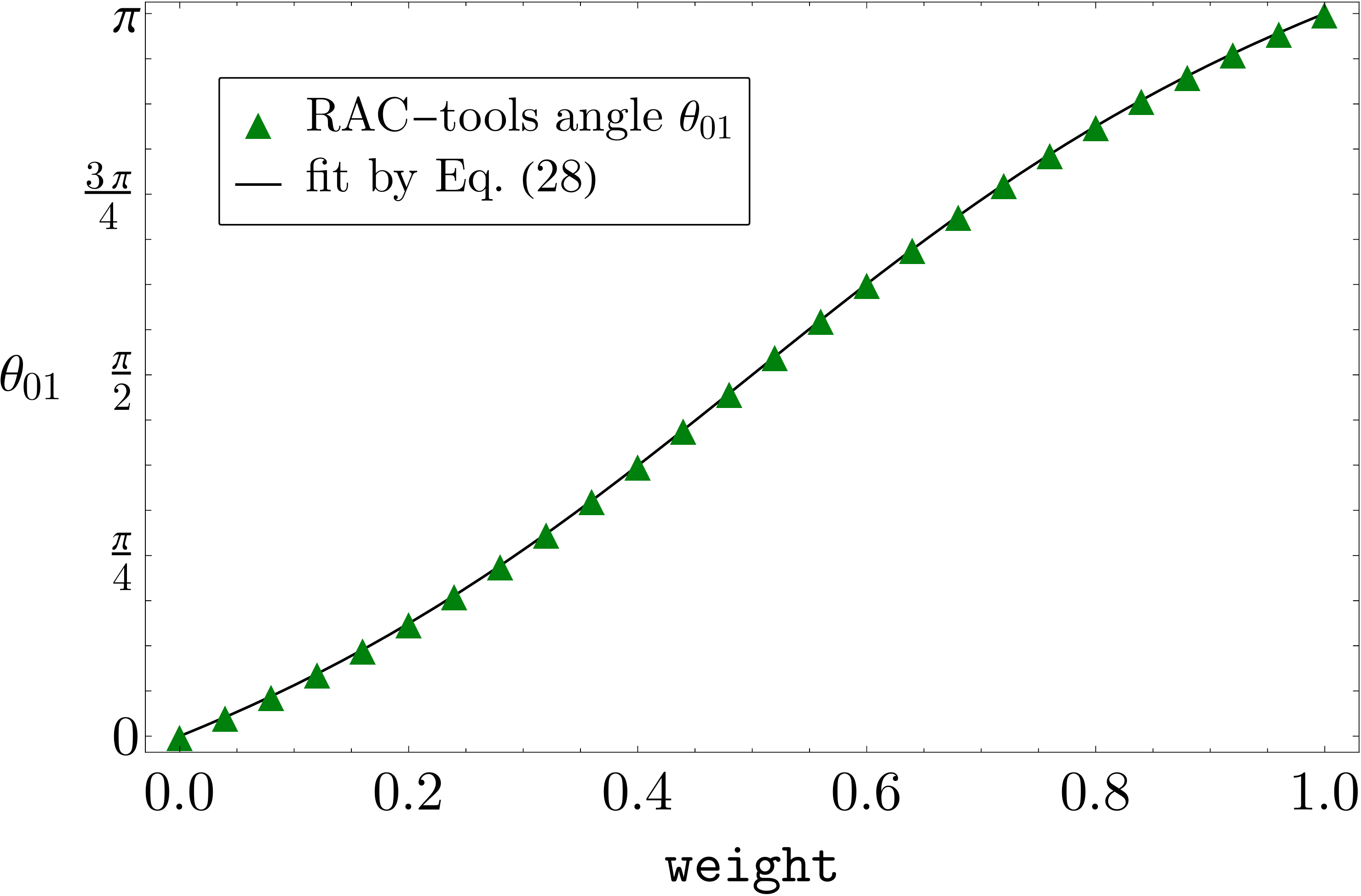}

    \caption{Top: Optimal performance of the $2^2\mapsto 1$ $b$-RAC with \protect\Verb+X\_CHESS+ bias, in combination with $r_0=r_1=\frac{1}{2}$, over quantum strategies (blue dots) and their classical counterpart (red squares). The numerical results for the quantum value are seen to agree with the theoretical prediction (solid line) extracted from Lemma \ref{lemma_2_bit}. Bottom: Angle between the Bloch vectors defining the optimal decoding strategy, as a function of the biasing parameter $w$. It is seen that these vectors are aligned only for the extremal values of $w$, meaning that the quantum value is strictly greater than the classical one for $w\in (0, 1)$.
    \label{fig:X_CHESS_2} }
\end{figure}

In the case $n=3$, as explained above, the value of both classical and quantum strategies becomes insensitive to variations of the biasing parameter $w$, coinciding with that of the unbiased RAC, which is given by
\begin{equation}
    F_Q=\frac{1}{2}+\frac{1}{2}\sqrt{\sum_y r_y^2}, 
\end{equation}
as follows from the previous discussion. It is perhaps surprising that even though in the extreme cases of $w \in \{0, 1\}$ we are left with only four out of the original eight strings, the RAC task does not get any easier. This could serve as an indication that some subsets of strings are as difficult to compress as the set of all strings (regardless of whether the compression is classical or quantum).

\subsection{The \texorpdfstring{\protect\Verb+X\_PLANE+}{} bias family}

So far, we have been introducing biases in the distribution of input strings by setting $\alpha_{{\bf x} y}=\alpha_{\bf x} r_y$, where $\alpha_{\bf x}$ is interpreted as the probability for input $\bf x$ to be encoded. As a result, when considering the individual random variables corresponding to the characters in the input strings, they will, in general, exhibit correlations. We can consider the case in which the input string characters are independently biased by defining $\alpha_{\bf x}=\prod_i \alpha_{x_i}$, where $\alpha_{x_i}$ denotes the probability of the $i$-th character in the string being $x_i$. In particular, for the case of two-bit input strings the bias tensor is given by $\alpha_{x_0x_1 y} =\frac{1}{2} \alpha_{x_0} \alpha_{x_1}$. 

Now let $\alpha_{x_1=0} = \alpha_{x_1=1} = 1/2$, and keep $\alpha_{x_0}$ arbitrary. It then follows that $p_{00}=\alpha_{00}+\alpha_{11}=\frac{1}{2}$, and $p_{01} = \alpha_{01} + \alpha_{10} = \frac{1}{2}$, implying that the bias has no net effect on the value of strategies that do not drop bits, as can be checked directly from Eq.~\eqref{gsqrt}. Consequently, the optimal value of this $2^2\mapsto 1$ $b$-RAC among quantum, non bit-dropping strategies, is $F_Q=\frac{1}{2}(1+\frac{1}{\sqrt{2}})$. We should note, however, that for extreme biases, the first character of the input string is either always $0$ or always $1$. In a situation as such there is no reason to include the first bit in the strategy, since the best performance can be obtained with a constant decoding function. Because this bit dropping could become optimal as a strategy for biases below the extremal value, in order to compute the quantum value we should compare $F_Q$ with the best value attained by a bit-dropping strategy, which is easily found to be
\begin{equation}
\FC_Q^1=\begin{cases}
    \frac{1}{2}(1+\alpha_{x_0=0}) \quad \text{if}\; \alpha_{x_0=0}>\alpha_{x_0=1} \\
    \frac{1}{2}(1+\alpha_{x_0=1}) \quad \text{otherwise},
\end{cases} \label{Fq1}   
\end{equation}
where the symbol $\FC_Q^1$, as defined in Eq.~\eqref{Fqs}, denotes that the strategy attaining this value does not encode one of the bits of the input string. Note that $\FC_Q^1$ coincides with the classical value for this RAC, since by ignoring a bit we are left with only one to consider in the encoding strategy. It follows from Eq.\eqref{Fq1} that dropping the first bit becomes optimal whenever $\FC_Q^1>\smash{\frac{1}{2}(1+\frac{1}{\sqrt{2}})}$, which occurs for $\alpha_{x_0}>\smash{\frac{1}{\sqrt{2}}}$ if $\alpha_{x_0}>\alpha_{x_1}$. In that case, the quantum value is therefore given by
\begin{equation}
    \FC_Q=\begin{cases}
           \frac{1}{2}(1+\frac{1}{\sqrt{2}}) & \text{if}\; \alpha_{x_0}\leq \frac{1}{\sqrt{2}} \alpha_{} \\
           \frac{1}{2}(1+\alpha_{x_0}) & \text{otherwise}
          \end{cases}
\end{equation}
FIG.~\ref{fig:x_plane2} depicts the numerical results provided by the RAC-tools package, which agree with the analytical value provided above.
\begin{figure}[b!]
    \centering
    \hfill\includegraphics[width = 8.65 cm]{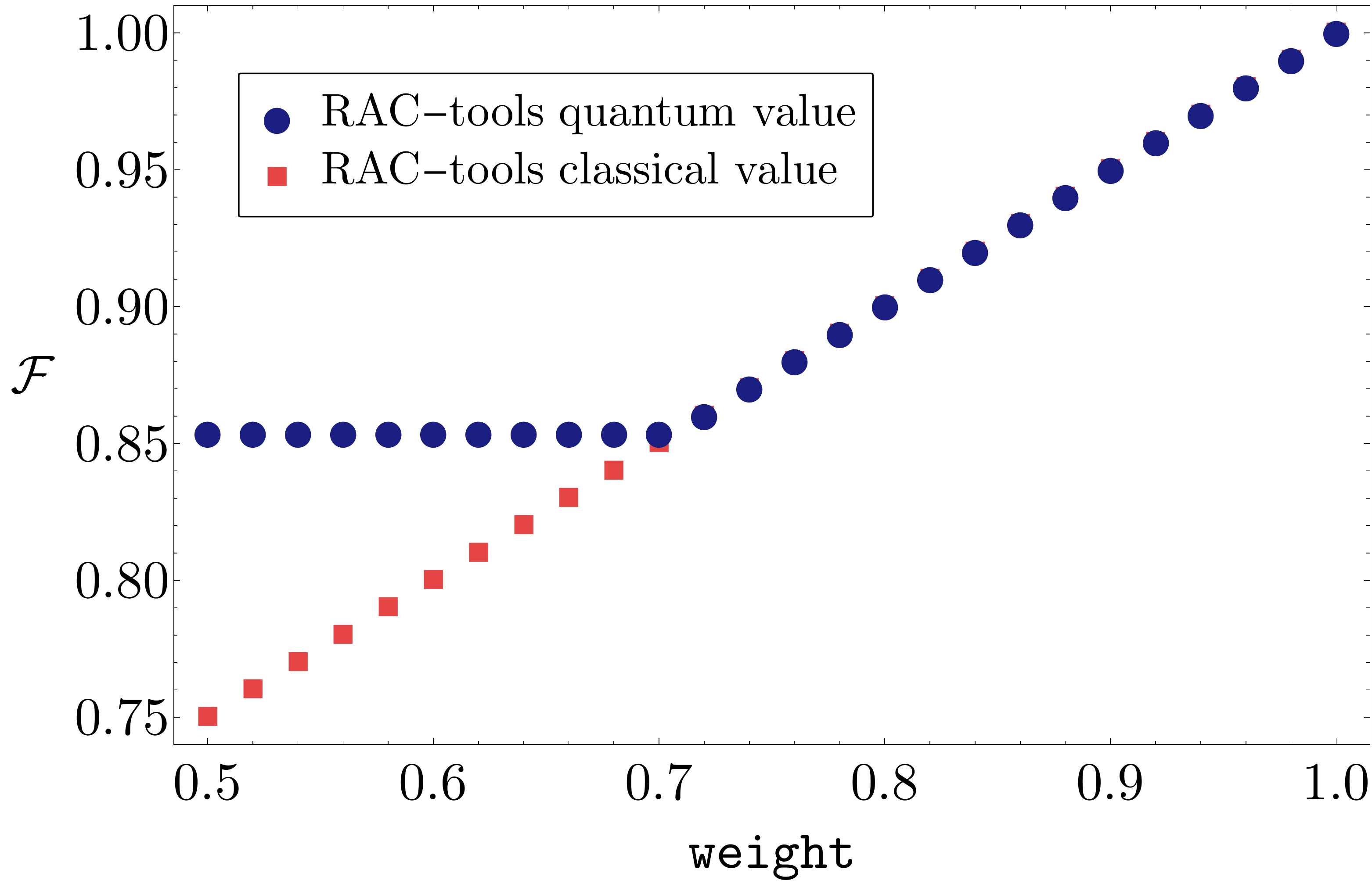}
    \caption{Optimal performance of the $2^2\mapsto 1$ $b$-RAC defined by the bias tensor $\alpha_{{\bf x} y}=\frac{1}{4}\alpha_{x_0}$. For $\alpha_{x_0}\leq\frac{1}{\sqrt{2}}$, the optimal strategy encodes both bits in the input string, and the quantum value (blue dots) is $\FC_Q=F_Q=\frac{1}{2}+(2\sqrt{2})^{-1}$ since the bias has no effect on the functional value. For $\alpha_{x_0}>\smash{\frac{1}{\sqrt{2}}}$, the best strategy does not encode the first bit, which reduces the value of the functional to the maximum attainable with classical strategies (red squares). For the region in which both bits are encoded, the angles representing the optimal measurements, as obtained by the RAC-tools package, are all $\frac{\pi}{2}$.} \label{fig:x_plane2}
\end{figure}

\begin{figure}[b!]
    \centering
    \hfill\includegraphics[width = 8.65 cm]{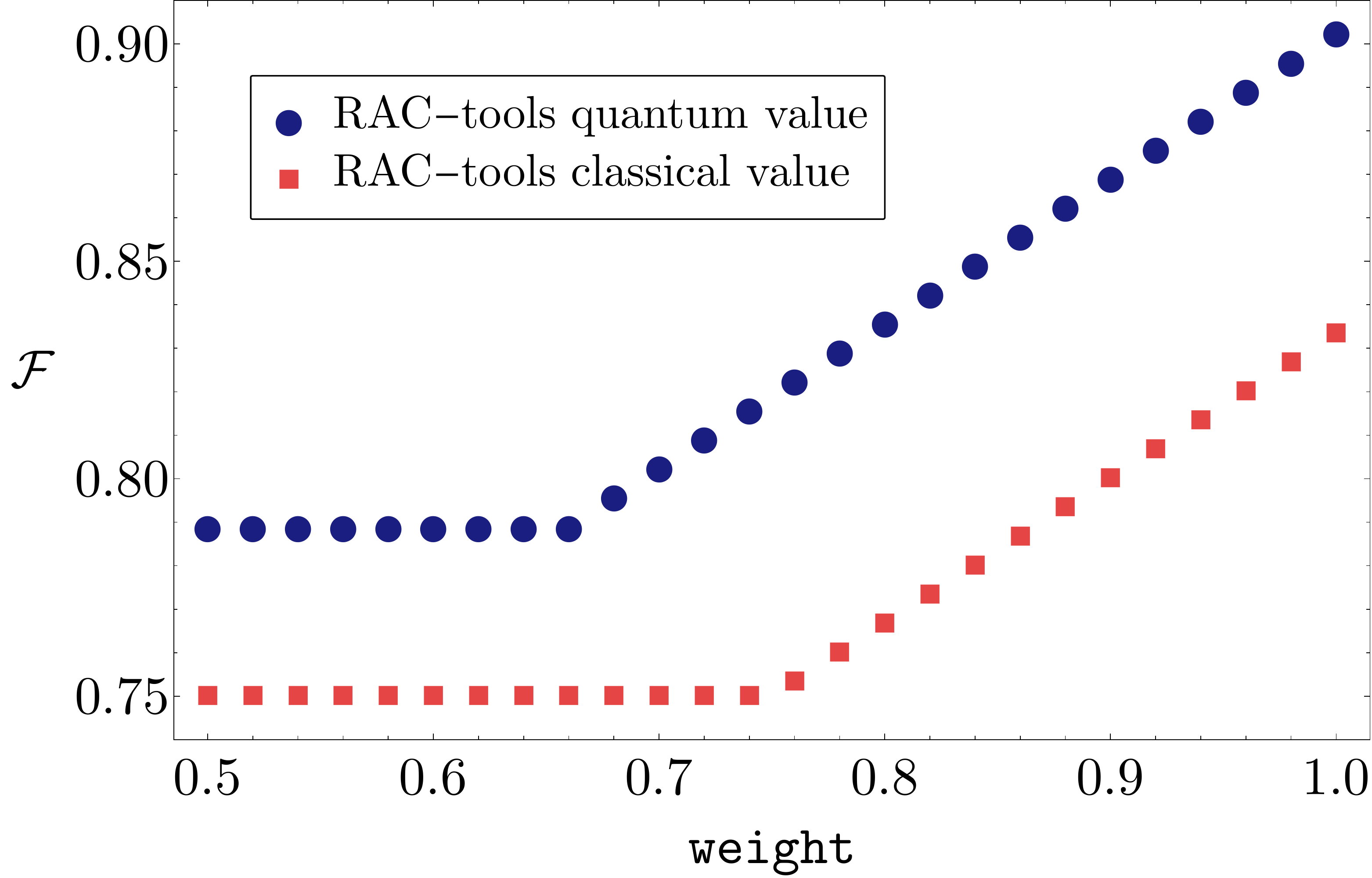}
    \caption{ Optimal performance of the $3^2\mapsto 1$ $b$-RAC defined by the bias tensor $\alpha_{x_0x_1x_2} = \smash{\frac{1}{4}}\alpha_{x_0}$. As observed in the $n=2$ case, there is a threshold value, $\alpha_{x_0} = \frac{1}{2}(1 + \sqrt{3} - \sqrt{2})$, above which the optimal quantum strategy ignores the first bit. Below the threshold the quantum value (blue dots) is $\FC_Q=F_Q=\frac{1}{2}(1+\frac{1}{\sqrt{3}})$, which coincides with the value of the unbiased RAC. For stronger biases, it becomes convenient to ignore the first bit, in which case the quantum value can be written as a shifted and rescaled version of the unbiased $2^2\mapsto 1$ RAC. As a result, unlike in the $n=2$ case, the optimal value above the threshold does not coincide with the classical value (red squares). Similarly to the scenario illustrated in FIG.~\ref{fig:x_plane2}, in the region where there are no ignored bits, the obtained optimal measurements are mutually orthogonal. When $x_0$ is ignored, measurement 0 becomes proportional to identity, while measurements 1 and 2 remain orthogonal.} \label{fig:x_plane3}
\end{figure}

A completely analogous analysis can be carried for the $n=3$ case for $\alpha_{x_0x_1x_2}=\frac{1}{4}\alpha_{x_0}$. As was the case for $n=2$, in the vicinity of the uniform distribution we expect the optimal strategy to include all three bits in the input string, in which case a direct calculation shows that the bias has no effect on the quantum value, which is given by $F_Q=\frac{1}{2}(1+\frac{1}{\sqrt{3}})$. As before we expect the optimal strategy to ignore the first bit when its value is strongly biased towards $0$ or $1$. If we take $\alpha_{x_0=0}\geq \alpha_{x_0=1}$, a direct calculation shows that the maximum value attained by a strategy ignoring the first bit is given by
\begin{equation}
    \FC_Q^1=\frac{1}{3}\left[\alpha_{x_0=0}+1+\frac{1}{\sqrt{2}}\right],
\end{equation}
which becomes larger than $F_Q$ for $\alpha_{x_0}>\frac{1}{2}(1+\sqrt{3}-\sqrt{2})$. The quantum value of this $b$-RAC is therefore given by the following piece-wise function
\begin{equation}
    \FC_Q=\begin{cases}
           \frac{1}{2}(1+\frac{1}{\sqrt{3}}) & \text{if}\; \alpha_{x_0}\leq \frac{1}{2}(1+\sqrt{3}-\sqrt{2}), \\
           \frac{1}{3}[\alpha_{x_0=0}+1+\frac{1}{\sqrt{2}}] & \text{otherwise.}
           \end{cases}
\end{equation}
FIG.~\ref{fig:x_plane3} depicts the results on the quantum and classical value of this $b$-RAC produced by the RAC-tools package, in full agreement with the analytical results provided above.

\subsection{The \texorpdfstring{\protect\Verb+B\_ONE+}{} bias family}

As described above, the elements of a bias tensor in this family take the form $\alpha_{\mathbf{x} y} = \frac{1}{n}\frac{1}{m^{n - 1}} \, w_{x_y}$. In the particular case of $n=2$ and $d=m=2$, which is addressed in Ref.~\cite{Kaczynska21}, we have
\begin{align}
\alpha_{\mathbf{x} y} = \frac{w_{x_y}}{4},
\end{align}
and a direct calculation using Eq.~\eqref{bitvalue} shows that, for this bias tensor, the $b$-RAC functional reads
\begin{equation}
    \FC_\leq\frac{1}{2}+ \frac{1}{4} \left[ \left|\cos\left(\frac{\theta}{2}\right)\right| + \sqrt{1-4\mu\cos^2\left(\frac{\theta}{2}\right)}\right] \label{b_one_221}
\end{equation}
with $\mu=w_0w_1$ and $\theta$ the angle between the Bloch vectors ${\bf m}_0$ and ${\bf m}_1$ characterizing the measurement operators. A search for critical points in Eq.~\eqref{b_one_221} shows that there are only two, satisfying either $\sin(\theta)=0$ or
\begin{equation}
    \cos\left(\frac{\theta}{2}\right)=\frac{1}{\sqrt{4\mu+16\mu^2}}. \label{b_one_angle}
\end{equation}
Whenever $\theta$ satisfies Eq.~\eqref{b_one_angle} the ensuing value of the functional reads
\begin{equation}
    \FC_Q=\frac{1}{2}+\frac{1}{8\sqrt{1+4\mu^2}}\left( \frac{1}{\sqrt{2}}+4\sqrt{\mu} \right), \label{b_one_qvalue}
\end{equation}
as previously reported in \cite{Kaczynska21}. As observed in the previous cases the condition $\sin(\theta)=0$ corresponds to commuting measurement operators, implying that in this case the value coincides with the classical value. The region in which this is the case is easily found to be $w_{0}\in[0, \frac{3-\sqrt{5}}{4}]\cup [\frac{1+\sqrt{5}}{4}, 1]$. In FIG.~\ref{fig:b_one_221}, both the quantum and classical values found by the RAC-tools package are depicted, and compared with $\FC_Q$ in Eq.~\eqref{b_one_qvalue}.

\begin{figure}[ht!]
    \centering
    \hfill\includegraphics[width = 8.6 cm]{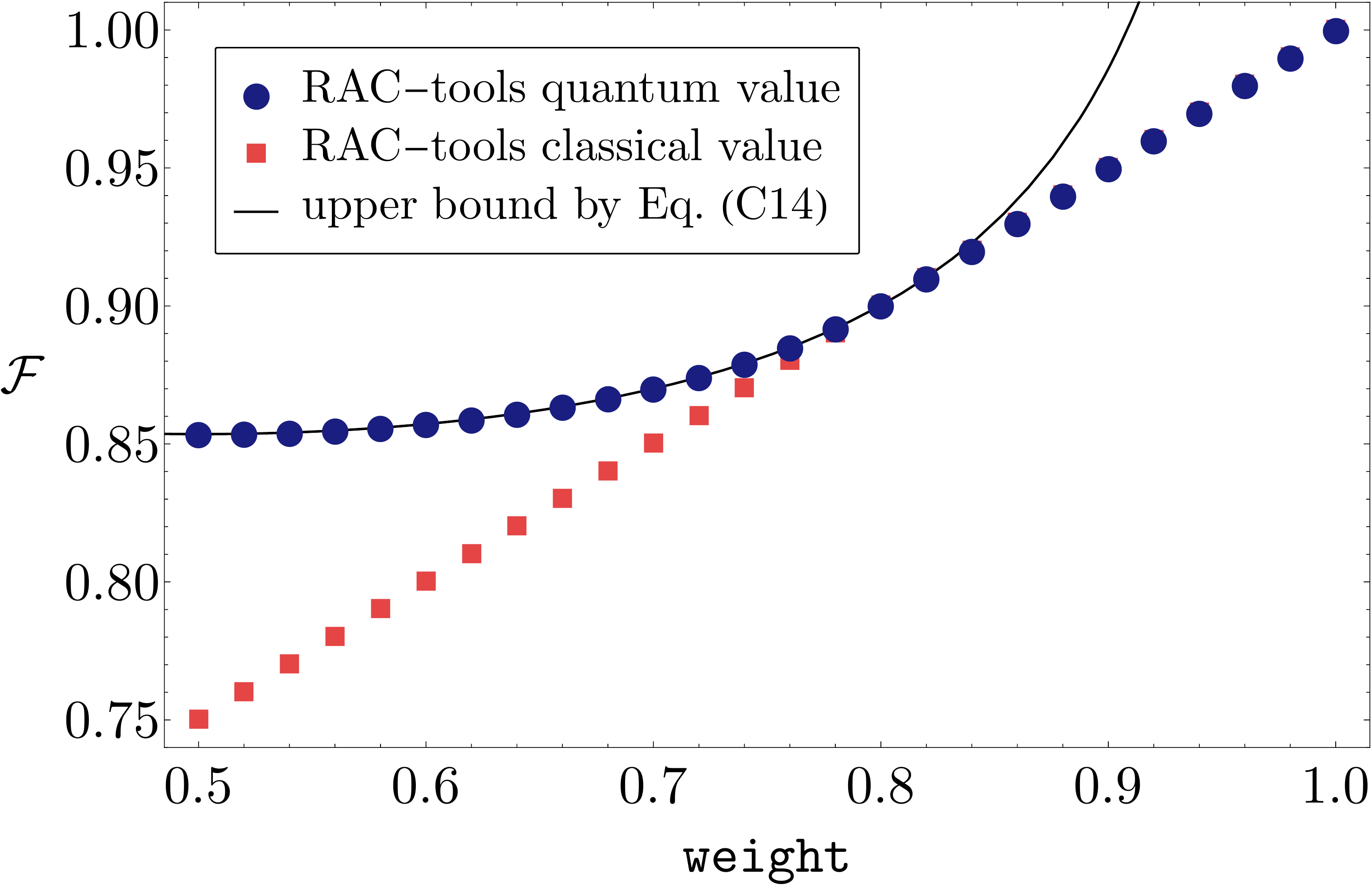}\\
    \hfill\includegraphics[width = 8.25 cm]{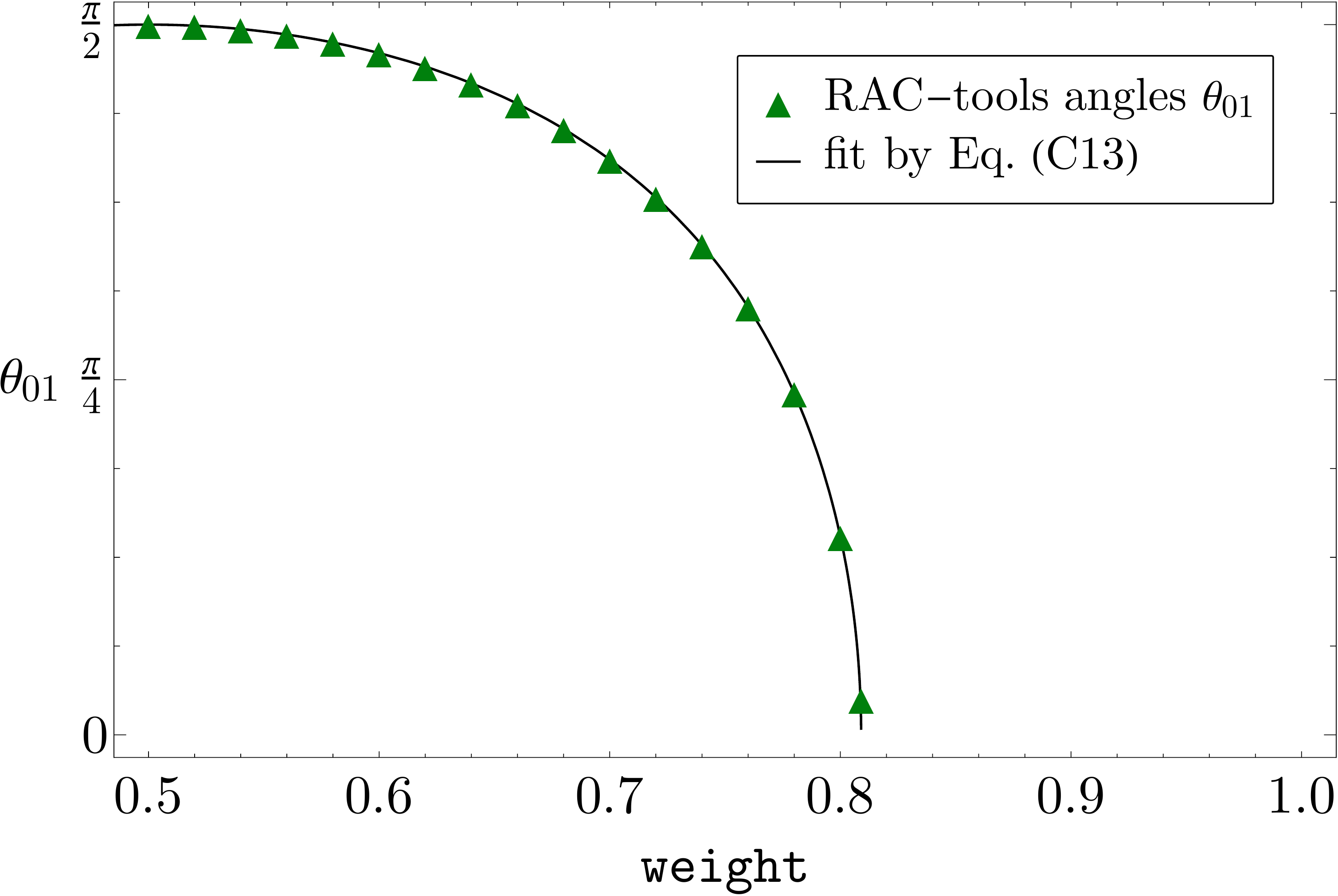}
    \caption{Top: Numerical quantum value (blue dots) of the of the biased RAC defined by $\alpha_{{\bf x} y}=\frac{w_{x_y}}{4}$, as a function of parameter $w_0$, compared to the corresponding classical value (red squares) and the value in Eq.~\eqref{b_one_qvalue} (solid line). Quantum strategies provide an advantage for $w_0\leq \frac{1+\sqrt{5}}{4}$. Bottom: Angle $\theta$ between the Bloch vectors parametrizing the optimal measurements. For $w_0$ within the region of quantum advantage, $\theta$ is given by Eq.~\eqref{b_one_angle}, becoming $0$ outside of it, where $\FC_Q=\FC_C$.}
    \label{fig:b_one_221}
\end{figure}

\section{Random biases for the \texorpdfstring{$2^{\MakeLowercase{d}}$}{} \texorpdfstring{$\mapsto 1$}{} RAC} \label{Ap:D}

In this appendix, we provide numerical evidence to support the claim that optimizing over projective measurements should be enough to find the quantum value of $2^d \mapsto 1$ $b$-RACs. Briefly, we tried to find counter-examples of such $b$-RACs in which the optimal realization is achieved by nonprojective measurements. In order to do that, we exhaustively evaluated the \verb+perform_search+ function for $2^d \mapsto 1$ $b$-RACs with $d \le 6$. We considered two kinds of biases: fully random biases and random factorizable biases. For the first case, we simply made up a bias tensor $\alpha_{\mathbf{x} y}$ the entries of which are uniformly distributed within the region $\alpha_{\mathbf{x} y} \ge 0$, for all $\mathbf{x}$, $y$, and $\sum_{\mathbf{x} y}\alpha_{\mathbf{x} y} = 1$. For the second case, we considered bias tensors such that $\alpha_{\mathbf{x} y} = \alpha_{\mathbf{x}} r_y$, where both $\alpha_x$ and $r_y$ are uniformly distributed over the inputs $\mathbf{x}$ and $y$, respectively.

The numerical results of this computation can be found in TABLE~\ref{projective vs not projective}. Apart from some pathological examples, all of the obtained realizations make use of projective measurements. This was the case for the vast majority of the computed samples except for five cases of random factorizable biases (four cases for $d=4$ and one case for $d=3$) in which the random draw of $r_y$ was almost deterministic. For these cases, since the weight in one of the measurements is almost zero, the classical and quantum values are numerically very close and due to insufficient numerical precision the optimization terminates with a nonprojective quantum strategy. 

\newpage

\begin{table*}[h!]
    \centering
    \caption{Samples with random biases for $2^d \mapsto 1$ RACs. This table consists of a compilation of numerical results produced by the RAC-tools package. The first column specifies the integer $d$ used for a given $2^d \mapsto 1$ RAC. The second column specifies one of the two kinds of biases used to generate this data set. The third column specify how many of such samples we considered for each case. In the two last columns, we show how often we retrieve a realization the measurements of which are all projective (fourth column) or not all projective (fifth column). The number of seeds used for each sample was 3.}
    \begin{tabular}{||c c c c c||}
        \hline
        \,$d$ & \protect\Verb+bias+ & No. samples & No. P & No. NP \\ [0.5 ex] 
        \hline\hline
        \,2 & \hspace{1mm} full random & 10\:000 & 10\:000 & 0\\ [0.1 ex]
        \hline
        \,2 & \hspace{1mm} factorizable random & 10\:000 & 10\:000 & 0\\ [0.1 ex]
        \hline
        \,3 & \hspace{1mm} full random & 10\:000 & 10\:000 & 0\\ [0.1 ex]
        \hline
        \,3 & \hspace{1mm} factorizable random & 10\:000 & 9\:999 & 1\\ [0.1 ex]
        \hline
        \,4 & \hspace{1mm} full random & 5\:000 & 5\:000 & 0 \\ [0.1 ex]
        \hline
        \,4 & \hspace{1mm} factorizable random & 5\:000 & 4\:996 & 4 \\ [0.1 ex]
        \hline
        \,5 & \hspace{1mm} full random & 2\:500 & 2\:500 & 0\\ [0.1 ex]
        \hline
        \,5 & \hspace{1mm} factorizable random & 2\:500 & 2\:500 & 0\\ [0.1 ex]
        \hline
        \,6 & \hspace{1mm} full random & 1\:000 & 1\:000 & 0 \\ [0.1 ex]
        \hline
        \,6 & \hspace{1mm} factorizable random & 1\:000 & 1\:000 & 0 \\ [0.1 ex]
        \hline
    \end{tabular}
    \label{projective vs not projective}
\end{table*}

\end{document}